\documentclass[11pt]{article}
\usepackage[margin=1in]{geometry}

\pdfoutput=1
\usepackage[protrusion=true,expansion=true]{microtype}
\usepackage{amsmath,amssymb,amsfonts,amsthm}
\usepackage{subcaption}
\usepackage{graphicx}
\usepackage{setspace}
\usepackage[backref=page]{hyperref}
\usepackage{color}
\usepackage{wrapfig}
\usepackage{tikz}
\usetikzlibrary{decorations.pathreplacing}
\usepackage{algorithm}
\usepackage{algorithmic}
\usepackage[framemethod=tikz]{mdframed}
\usepackage{xspace}
\usepackage{pgfplots}
\usepackage{framed}
\usepackage{thmtools}
\usepackage{thm-restate}
\pgfplotsset{compat=1.5}

\newtheorem{theorem}{Theorem}[section]

\newtheorem{lemma}[theorem]{Lemma}
\newtheorem{proposition}[theorem]{Proposition}
\newtheorem{definition}[theorem]{Definition}

\newenvironment{proofof}[1]{\begin{trivlist} \item {\bf Proof
#1:~~}}
  {\qed\end{trivlist}}

\newcommand{\namedref}[2]{\hyperref[#2]{#1~\ref*{#2}}}
\newcommand{\thmlab}[1]{\label{thm:#1}}
\newcommand{\thmref}[1]{\namedref{Theorem}{thm:#1}}
\newcommand{\lemlab}[1]{\label{lem:#1}}
\newcommand{\lemref}[1]{\namedref{Lemma}{lem:#1}}

\newcommand{\seclab}[1]{\label{sec:#1}}
\newcommand{\secref}[1]{\namedref{Section}{sec:#1}}
\newcommand{\applab}[1]{\label{app:#1}}
\newcommand{\appref}[1]{\namedref{Appendix}{app:#1}}

\newcommand{\alglab}[1]{\label{alg:#1}}
\newcommand{\algref}[1]{\namedref{Algorithm}{alg:#1}}

\newcommand{\deflab}[1]{\label{def:#1}}

\def \MSE    {\mdef{\mathsf{MSE}}}

\newcommand\norm[1]{\left\lVert#1\right\rVert}
\newcommand{\PPr}[1]{\ensuremath{\mathbf{Pr}\left[#1\right]}}
\newcommand{\PPPr}[2]{\ensuremath{\underset{#1}{\mathbf{Pr}}\left[#2\right]}}
\newcommand{\Ex}[1]{\ensuremath{\mathbb{E}\left[#1\right]}}
\newcommand{\Var}[1]{\ensuremath{\mathbb{V}\left[#1\right]}}

\renewcommand{\O}[1]{\ensuremath{\mathcal{O}\left(#1\right)}}
\newcommand{\tO}[1]{\ensuremath{\tilde{\mathcal{O}}\left(#1\right)}}
\newcommand{\eps}{\varepsilon}

\def \calA    {\mdef{\mathcal{A}}}

\def \calD    {\mdef{\mathcal{D}}}

\def \calM    {\mdef{\mathcal{M}}}

\def \calP    {\mdef{\mathcal{P}}}

\def \calR    {\mdef{\mathcal{R}}}
\def \calS    {\mdef{\mathcal{S}}}

\def \calV    {\mdef{\mathcal{V}}}

\def \calZ    {\mdef{\mathcal{Z}}}

\newcommand{\mdef}[1]{{\ensuremath{#1}}\xspace}  

\DeclareMathOperator*{\Ber}{Ber}
\DeclareMathOperator*{\Unif}{Unif}

\newcommand{\E}[2][]{\mdef{\underset{#1}{\mathbb{E}}\left[#2\right]}} 

\newcommand{\ignore}[1]{}

\newif\ifnotes\notestrue 
\ifnotes

\newcommand{\samson}[1]{\textcolor{purple}{{\bf (Samson:} {#1}{\bf ) }} \marginpar{\tiny\bf
             \begin{minipage}[t]{0.5in}
               \raggedright S:
            \end{minipage}}}            							
\else
\newcommand{\samson}[1]{}
\fi

\makeatletter
\renewcommand*{\@fnsymbol}[1]{\textcolor{mahogany}{\ensuremath{\ifcase#1\or *\or \dagger\or \ddagger\or
 \mathsection\or \triangledown\or \mathparagraph\or \|\or **\or \dagger\dagger
   \or \ddagger\ddagger \else\@ctrerr\fi}}}
\makeatother

\providecommand{\email}[1]{\href{mailto:#1}{\nolinkurl{#1}\xspace}}

\definecolor{mahogany}{rgb}{0.75, 0.25, 0.0}
\definecolor{darkblue}{rgb}{0.0, 0.0, 0.55}
\definecolor{darkpastelgreen}{rgb}{0.01, 0.75, 0.24}
\definecolor{bleudefrance}{rgb}{0.19, 0.55, 0.91}
\definecolor{darkgreen}{rgb}{0.0, 0.2, 0.13}
\definecolor{darkgoldenrod}{rgb}{0.72, 0.53, 0.04}
\definecolor{darkred}{rgb}{0.55, 0.0, 0.0}
\hypersetup{
     colorlinks   = true,
     citecolor    = blue,
		 linkcolor		= blue
}

\usepackage{comment}
\usepackage{mathtools}

\usepackage[capitalize]{cleveref}
\newcommand{\Err}{\mathsf{Err}}
\renewcommand{\E}{\mathbb{E}}
\newcommand{\sphere}{\mathbb{S}}
\newcommand{\dist}{\mathsf{dist}}

\newcommand{\reals}{\mathbb{R}} 
\newcommand{\ball}{\mathbb{B}}
\newcommand{\ltwo}[1]{\norm{#1}_2} 
\newcommand{\defeq}{\coloneqq}

\usepackage{etoolbox}

\newtoggle{arxiv}
\toggletrue{arxiv}

\newif\ificml

\title{Private Vector Mean Estimation in the Shuffle Model: \\Optimal Rates Require Many Messages}
\author{Hilal Asi\thanks{Apple Inc. \texttt{hilal.asi94@gmail.com}.} \and Vitaly Feldman\thanks{Apple Inc. \texttt{vitaly.edu@gmail.com}.} \and Jelani Nelson\thanks{UC Berkeley. \texttt{minilek@berkeley.edu}. Supported by NSF grant CCF-1951384, ONR grant N00014-18-1-2562, and ONR DORECG award N00014-17-1-2127.} \and Huy L. Nguyen\thanks{Northeastern. \texttt{huylenguyen@gmail.com}.  Supported in part by NSF CAREER grant CCF-1750716 and NSF CCF-2311649.} \and 
Kunal Talwar\thanks{Apple Inc. \texttt{kunal@kunaltalwar.org}.} \and
Samson Zhou\thanks{Texas A\&M University. \texttt{samsonzhou@gmail.com}. Supported in part by NSF CCF-2335411.}}

\begin{document}
\allowdisplaybreaks
\maketitle

\begin{abstract}
We study the problem of private vector mean estimation in the shuffle model of privacy where $n$ users each have a unit vector $v^{(i)} \in\mathbb{R}^d$. We propose a new  multi-message protocol that achieves the optimal error using $\tilde{\mathcal{O}}\left(\min(n\varepsilon^2,d)\right)$ messages per user. Moreover, we show that any (unbiased) protocol that achieves optimal error requires each user to send $\Omega(\min(n\varepsilon^2,d)/\log(n))$ messages, demonstrating the optimality of our message complexity up to logarithmic factors.

Additionally, we study the single-message setting and design a protocol that achieves mean squared error $\O{dn^{d/(d+2)}\varepsilon^{-4/(d+2)}}$. Moreover, we show that \emph{any} single-message protocol must incur mean squared error $\Omega(dn^{d/(d+2)})$, showing that our protocol is optimal in the  standard setting where $\varepsilon = \Theta(1)$. Finally, we study robustness to malicious users and show that malicious users can incur large additive error with a single shuffler.

\end{abstract}

\section{Introduction}
Vector mean estimation is a fundamental problem in federated learning, where a large number of distributed users can provide information to collaboratively train a machine learning model. 
Formally, there are $n$ users that each have a real-valued vector $v^{(i)}\in\mathbb{R}^d$. In the vector mean estimation problem, the goal is to compute the average of the vectors $v=\frac{1}{n}\sum_{i=1}^n v^{(i)}$, whereas in the closely related vector aggregation problem, the goal is to compute the sum of the vectors $nv = \sum_{i=1}^n v^{(i)}$. As the privacy error scales with the norms of the vectors, we normalize and thus assume that $\|v^{(i)}\|_2 \leq 1$.
The vectors could represent frequencies of sequences of words in smartphone data for predictive text suggestions, shopping records for financial transactions or recommendation systems, various medical statistics for patients from different healthcare institutions, or gradient updates to be used to train a machine learning model. 
Thus, vector mean estimation and vector aggregation are used in a number of applications, such as deep learning through federated learning~\cite{ShokriS15,AbadiCGMMT016,McMahanMRHA17}, frequent itemset mining~\cite{sun2014personalized}, linear regression~\cite{NguyenXYHSS16}, and stochastic optimization~\cite{ChaudhuriMS11,CheuJMP22}. 

Due to the sensitive nature of many of these data types, recent efforts have concentrated on facilitating federated analytics while preserving privacy. 
Differential privacy (DP)~\cite{DworkMNS06} has emerged as a widely adopted rigorous mathematical definition that quantifies the amount of privacy leaked by a mechanism for any given individual user. 
In particular, local differential privacy (LDP)~\cite{KasiviswanathanLNRS11} demands that the distribution of the transcript of the communication protocol cannot be greatly affected by a change in a single distributed user's input. 
This approach enables the distributed collection of insightful statistics about a population, while protecting the private information of individual data subjects even with an untrusted curator who analyzes the collected statistics. 

Unfortunately, in order to ensure privacy, the local model often requires a high amount of noise that results in poor accuracy of the resulting mechanisms. 
For example, in the simple case where $v^{(i)}\in\{0,1\}$, i.e., binary summation, there exist private mechanisms with $\O{1}$ additive error in the central setting where the data curator is trusted~\cite{DworkMNS16}, but the additive error must be $\Omega(\sqrt{n})$ in the local model~\cite{BeimelNO08,ChanSS12}. 
Consequently, the Encode, Shuffle, Analyze (ESA) model was proposed as an alternative distributed setting that could potentially result in a lower error~\cite{BittauEMMRLRKTS17}. The shuffle model of privacy
is a special case of the ESA framework introduced by~\cite{CheuSUZZ19}, where a trusted shuffler receives and permutes a set of encoded messages from the distributed users, before passing them to an untrusted data curator. 
\cite{CheuSUZZ19} and \cite{ErlingssonFeMiRaTaTh19} showed that for the important tasks  of binary and real-valued summation, there are shuffle protocols that nearly match the accuracy of the optimal central DP mechanisms. 
Of note, \cite{BalleBGN19,BalleBGN20,GhaziMPV20,GhaziKMPS21} study the $1$-dimensional real summation problem both under the lens of minimizing the error and the message complexity to achieve optimal error. 
In particular, \cite{GhaziKMPS21} show that there is an optimal protocol that requires each user to send $1+o(1)$ messages in expectation. However, the natural extension of their approach to $d$-dimensional mean estimation requires a number of messages that is exponential in $d$.

For $d$-dimensional mean estimation in the single-message setting, the most relevant works are that of \cite{ScottCM21,ScottCM22}, who study minimizing the mean-squared error of protocols that aim to compute the mean of vectors $u^{(1)},\ldots,u^{(n)}\in\mathbb{R}^d$, where each sampled vector $u^{(i)}$ consists of a number of coordinates sampled from the input vector $v^{(i)}\in\mathbb{R}^d$. 
\cite{ScottCM21,ScottCM22} treat the sampled vectors $u^{(i)}$ as the true vectors and show a single-message shuffle protocol for estimating their mean. 
However, the mean-squared error of the overall protocol can be large, due to the large variance incurred by the procedure of sampling vectors $u^{(i)}$ from the true vectors. 

Private vector mean estimation in the shuffle model is thus not well-understood, both under single-message and multi-message settings. In particular the following natural questions are open: first, in the multi-message setting, what is the total number of messages required in the shuffle model in order to obtain optimal rates for vector mean estimation. Secondly, what are the optimal algorithms for the single-messages setting.

Another desiderata in the design of distributed algorithms is that of robustness to malicious agents. 
In our context, we would like the system to be somewhat robust to one or a small number of clients that behave maliciously. 
For a problem like vector aggregation, a client can always misrepresent their input, and thus impact the sum; when vectors are restricted to having norm at most one, this can impact the true sum by at most two in the norm. 
The poisoning robustness of a protocol is defined to be $\rho$ if the impact of an adversarial client on the computed sum is upper bounded, in Euclidean norm, by $\rho$. Thus a protocol that computes the exact sum has robustness $2$. We would like to design protocols with robustness that is not much larger. We note that robustness of this kind has been previously been studied in other models of privacy~\cite{CheuSU19,Talwar22}. In the shuffle model with multiple messages, there are two different possible models from the robustness point of view.

In any implementation of a shuffle protocol that aims to achieve robustness, one must limit the number of contributions a single client can make: indeed if a single malicious client can pretend to be a million different clients without being detected, one cannot hope to achieve any reasonable robustness. 
In typical implementations of  a shuffler, such control can be achieved. For example, in a mix-net implementation of shuffling~\cite{BittauEMMRLRKTS17}, each client sends a non-anonymous but encrypted message to the first hop, where this first server can see who sent the message but not the contents of the message. 
This first hop can then validate that each sender sends at most one, or at most a predetermined number of messages to the server. 
When this bound is $B$ and there are $n$ clients, this server can implement this rate control using $\O{n}$ counters that can count up to $B$, for a total of $n\log_2 (B+1)$ bits of storage. 
We call this model, where each client can send a bounded number of messages to a single shuffler, the {\em multi-message shuffle} model.

This is distinct from a {\em multi-shuffler} model, where a client is allowed to send $1$ message to each of $B$ shufflers (or equivalently, $B$ messages to a single shuffler with the constraint that there be at most $1$ of each of $B$ ``types'' of message). 
To ensure robustness, the shuffler would then need to rate-limit each type of message. 
When implemented in a mix-net setting as above, this multi-shuffler would require the first hop server to store $\O{nB}$ bits. 
It is easy to see that information-theoretically, a server cannot ensure $nB$ separate rate limits using $o(nB)$ bits of state. 
For large $B$, this is significantly more than the  $n\log_2 (B+1)$ bits that suffice for the multi-message shuffle. 

Similarly in other implementations, e.g. those building on PrivacyPass or OHTTP tokens~\cite{DavidsonGSTV18, ietf-ohttp, ietf-privacypass-rate-limit-tokens-01}, there is a server that implements the rate control at some step, and its cost scales as $nB$ for multiple shufflers, compared to $n\log_2 (B+1)$ for a multi-message shuffle. Thus from an overhead point of view, these two models are significantly different. As a concrete example, when $n = 10^8$, and the vectors are $d=10^6$-dimensional, a $d$-message shuffle requires a few hundred megabytes of storage for the counters, whereas a multi-shuffle would require $12$ terabytes of storage. It is thus much preferable to design algorithms that are robust in the multi-message shuffle model, rather than in the multi-shuffler model.

\subsection{Our Contributions}
In this work, we study the vector aggregation and vector mean estimation problems in the shuffle model of privacy, both in the single-message and multi-message settings, and from the viewpoint of robustness.
We show the following results.

\paragraph{Multiple messages per user (\Cref{sec:multi-msg}).}
We consider the multi-message setting where users are allowed to send multiple messages.
We propose a new protocol in the shuffle model that obtains optimal mean squared error of $\tO{\frac{d}{\eps^2}}$ using $\tO{\min(d,n\eps^2)}$ messages per user, matching the performance of the central model of privacy~\cite{BassilyST14} up to logarithmic factors.

\begin{restatable}{theorem}{thmmanymsgub}
\thmlab{thm:many:msg:ub}
There exists an $(\eps,\delta)$-DP mechanism for vector aggregation that uses $\tO{\min(d,n\eps^2)}$ messages per user and achieves mean squared error $ \tO{\frac{d}{\eps^2}}$. 
\end{restatable}

Moreover, we prove the following lower bound which shows that $\Omega(\min(n\eps^2,d)/\log(n))$ messages are necessary in the shuffle model in order to obtain the optimal rate. The lower bound holds for any unbiased or summation protocol (as we define in~\Cref{sec:pre}).
\begin{restatable}{theorem}{thmmanymsglb}
\thmlab{thm:many:msg:lb}
For any (unbiased or summation) $(\eps,\delta)$-Shuffle DP protocol for vector aggregation that achieves the optimal mean squared error $\O{\frac{d}{\eps^2}}$, must send $k = \Omega(\min(n\eps^2,d)/\log(n))$ messages.
\end{restatable}

\paragraph{Single message per user (\Cref{sec:single-msg}).} 
We also study the single-message setting where each user is allowed to send only a single message. 
We show that there exists a private protocol that can achieve mean squared error $\O{dn^{d/(d+2)}\eps^{-4/(d+2)}}$.
\begin{restatable}{theorem}{thmonemsgub}
\thmlab{thm:one:msg:ub}
For any $\eps\in(0,1)$, $\delta\in(0,1)$, and $d,n\in\mathbb{N}$, there exists an $(\eps,\delta)$-DP protocol in the one-message shuffle model with mean squared error $\mathcal{O}_{\delta}\left(dn^{d/(d+2)}\eps^{-4/(d+2)}\right)$.
\end{restatable}
Though the mean squared error of \thmref{thm:one:msg:ub} seems somewhat arbitrary, we show that it is tight for a single message per user shuffle.
\begin{restatable}{theorem}{thmonemsglb}
Let $\calP$ be an $(\eps,\delta)$-DP protocol for vector aggregation on the unit ball $\ball^{d-1}_2$ in the one-message shuffle model with $\delta<\frac{1}{2}$. 
Then the mean squared error of $\calP$ satisfies $\MSE(\calP)=\Omega\left(dn^{d/(d+2)}\right)$.
\end{restatable}

\paragraph{Robustness to malicious users (\Cref{sec:rob}).}
We subsequently study the robustness of shuffle DP protocols to malicious users, who may distribute adversarial messages in an effort to induce the maximal possible mean squared error by a protocol. 

We first show that for additive protocols in the multi-message shuffle model, each malicious user can induce additive mean squared error up to $\Omega\left(\frac{d}{\log^2(nd)}\right)$, for a total of $\Omega\left(\frac{kd}{\log^2(nd)}\right)$ additive mean squared error across $k$ malicious users. 
More generally, we show the following result for the case of $s$ shufflers.
\begin{restatable}{theorem}{thmbaduserserr}
\thmlab{thm:bad:users:err}
Let $\eps=\O{1}$ and $\delta<\frac{1}{nd}$. 
Then any $(\eps,\delta)$-DP mechanism for vector summation in which $s$ shufflers take messages corresponding to a disjoint subset of the coordinates and returns the sum of the messages across $n$ players with $k$ malicious users has additive error mean squared error $\Omega\left(\frac{kd}{s\log^2(nd)}\right)$. 
\end{restatable}

On the other hand, we show that our protocol is robust to malicious users when multiple shufflers exist: in this case, $k$ malicious users can only induce error $\O{k}$, rather than $\Omega(kd)$. 
Since the input of each user is a vector with at most unit length, then our result essentially says that a malicious user can at most hide its input vector by generating the protocol for a different vector. 
By comparison, each malicious user in the context of \thmref{thm:bad:users:err} can be responsible for error $\Omega\left(\frac{d}{\log^2(nd)}\right)$, which can be significantly larger than the unit length of each input vector. 


Thus our results show that a large class of accurate protocols in the multi-message shuffle model are inherently non-robust. While the multi-shuffler model can allow for better robustness, it comes at a significant additional cost. We remark that an trusted {\em aggregator} such as one built on top of PRIO~\cite{Corrigan-GibbsB17} can ensure high robustness as well as low overhead (c.f.~\cite{RothblumOCT23}). While it is more complex to implement a trusted aggregator (compared to a shuffler), our results point to an important reason why a shuffler may not be sufficient when robustness is a concern.

\subsection{Related Work}
\label{sec:related}
Mean estimation is a fundamental problem for data analytics and is the building block for many algorithms in stochastic optimization such as stochastic gradient descent. 
As a result, privacy-preserving frequency estimation has been extensively studied in applications of federated learning. 

\paragraph{Real summation in the shuffle model.}
There has also been a line of work studying real summation, i.e., vector summation with $d=1$, in the shuffle model. 
In the single-message shuffle model, \cite{BalleBGN19} showed that the optimal additive error is $\tilde{\Theta}_{\eps}(n^{1/6})$, whereas in the multi-message shuffle model, there exist protocols that achieve additive error $\O{\frac{1}{\eps}}$~\cite{BalleBGN20,GhaziMPV20,GhaziKMPS21}. 
In particular, \cite{BalleBGN20,GhaziMPV20} use the split-and-mix protocol of \cite{IshaiKOS06} to achieve additive error $\O{\frac{1}{\eps}}$, though at the cost of using at least $3$ messages, each of length at least $\frac{\log\frac{1}{\delta}}{\log n}$. 
Subsequently, the protocol of \cite{GhaziKMPS21} achieves near-optimal error, while using only $1+o(1)$ messages per user in expectation.

\paragraph{Lower bounds for the multi-message shuffle model.} For the problem of mean estimation, existing work does not have any lower bounds in the multi-shuffle model. However, for other problems such as private selection or parity learning, several recent papers have demonstrated new lower bounds for the multi-message model~\cite{cheu2021limits,chen2020distributed,beimel2020round}. More precisely, \cite{cheu2021limits} proved new lower bounds of $\Omega(\sqrt{D})$ on the sample complexity of selection from $D$ candidates (and other learning problems) under the pan-privacy model, which implies lower bounds for the shuffle model. However, their results do not extend to our setting as high-dimensional mean estimation is not difficult in the pan-private model and thus do not translate to strong lower bounds for privacy in the shuffle model. Moreover, \cite{chen2020distributed} proved lower bounds of $D/k$ for private selection for the multi-message model with $k$ messages. Finally, \cite{beimel2020round} consider the the common element problem (which aims to identify an element that is common to all users) and prove non-trivial lower bounds for the multi-message model when $k$ is small. However, these lower bounds are different from ours in two distinct ways: first, none of them hold for the problem of high-dimensional mean estimation, and secondly, they do not exhibit the same phase transition behavior that our lower bounds show, where an optimal rate is achieved only when $k \ge d$. 


\paragraph{Mean estimation in the LDP model.}
\cite{DuchiR19,DuchiWJ16} studied the vector mean estimation problem in the LDP model, showing how to achieve optimal error without accounting for any communication constraints. 
\cite{BhowmickDFKR18} developed a new algorithm, PrivUnit, and proved it is optimal up to constants, and~\cite{AsiFeTa22} show that PrivUnit with optimized parameters is the optimal mechanism. More recently, there have been several works that study LDP aggregation with low communication cost such as~\cite{ChenKO20,FeldmanTa21,AsiFeNeNgTa23}. Another line of work considers improving the communication cost in the setting where the input vector of each user is $k$-sparse~\cite{BassilyS15,FantiPE16,YeB18,AcharyaS19,Zhou0CFS22}.

\subsection{Preliminaries and problem setting}
\label{sec:pre}
\paragraph{Notation.}
We let $\sphere^{d-1} = \{ v \in \reals^d : \ltwo{v} = 1 \}$ denote the $d$-dimensional sphere. For a set $S \subseteq \reals^d$ and a vector $v \in \reals^d$, define $\dist(v,S) = \inf_{u \in S} \ltwo{v-u}^2$. 
Let $\ball^{d-1} = \{ v \in \reals^d: \ltwo{v} \le 1\}$ be the unit ball in $d$ dimensions. 

We recall the standard definition of differential privacy.
\begin{definition}[Differential privacy]
\deflab{def:dp}
\cite{DworkMNS06}
Given a privacy parameter $\eps>0$ and a failure parameter $\delta\in(0,1)$, a randomized algorithm $\calA:\calD\to\calR$ is $(\eps,\delta)$-differentially private if, for every neighboring datasets $D,D'\in\calD$ and for all $U\subseteq\calR$,
\[\PPr{\calA(D)\in U}\le e^{\eps}\cdot\PPr{\calA(D')\in U}+\delta.\]
\end{definition}
We require the standard advanced composition of differential privacy. 
\begin{theorem}[Advanced composition of differential privacy~{\cite{DworkR14}}]
\thmlab{thm:dp:adv:comp}
Let $\eps,\delta\ge0$ and $\delta'>0$. 
The  composition  of $k$ algorithms that are each $(\eps,\delta)$-differentially private is itself $(\tilde{\eps},\tilde{\delta})$-differentially private, where
\[\tilde{\eps}=\eps\sqrt{2k\ln(1/\delta')}+k\eps\left(\frac{e^\eps-1}{e^{\eps}+1}\right),\qquad \tilde{\delta}=k\delta+\delta'.\]
\end{theorem}

\paragraph{Shuffle differential privacy.} 
In the shuffle DP model, we have $n$ users, each holding a vector $v_i \in \sphere^{d-1}$. A protocol in this model is a pair of procedures $(\calA,\calR)$ where $\calR: \sphere^{d-1} \to \calZ^k$ is a local randomizer that each user applies to produce $k$ messages in $\calZ$. 
Then, a shuffler $\Pi$ is applied to all messages output by the users, before applying an aggregation $\calA: \calZ^\star \to \reals^d$ over the shuffled messages to return an output 
\begin{equation*}
\hat v = \calA(\Pi(\calR(v_1),\dots,\calR(v_n)).
\end{equation*}
We say that a protocol $(\calA,\calR)$ is \emph{$(\eps,\delta)$-Shuffle DP} if the algorithm that outputs $\Pi(\calR(v_1),\dots,\calR(v_n))$ is $(\eps,\delta)$-DP. 
Moreover, we define the mean squared error $\Err(\calA,\calR)$ of the protocol to be
\[\Err(\calA,\calR) \defeq \sup_{v_1,\dots,v_n \in \sphere^{d-1}} \E \left[\norm{\calA(\Pi(\calR(v_1),\dots,\calR(v_n))) - \sum_{i=1}^n v_i}_2^2 \right].\]

Throughout the paper, we use the notion of \emph{unbiased} and \emph{summation} protocols. More specifically, we say that a protocol $(\calA,\calR)$ is unbiased if for all $v_1,\dots,v_n \in \sphere^{d-1}$
\begin{equation*}
\E[\calA(\Pi(\calR(v_1),\dots,\calR(v_n))] = \sum_{i=1}^n v_i.
\end{equation*}
We also let $\calA^+$ denote the summation aggregation, that is, 
\begin{equation*}
    \calA^+(\Pi(\calR(v_1),\dots,\calR(v_n))) = \sum_{m \in \Pi(\calR(v_1),\dots,\calR(v_n)) } m.
\end{equation*}
These notions will be useful for our lower and upper bounds.

\subsubsection{Kashin representation}
\newcommand{\Ck}{\mathsf{C}_\mathsf{K}}
We use Kashin's representation in our multi-message algorithms, which has the following property.
\begin{lemma}[Kashin's representation]
\cite{LyubarskiiV10}
Let $d \ge 1$.
There exists a transformation $U_{\mathsf{K}} \in \reals^{2d \times d}$ and a constant $\Ck$ such that 
\begin{enumerate}
\item 
$U_{\mathsf{K}}^T U_{\mathsf{K}}= I_d$
\item
For all $x \in \sphere^{d-1}$, $\|U_{\mathsf{K}} x\|_\infty\le\frac{\Ck}{\sqrt{d}}$. 
\end{enumerate}
Here we use the subscript $\mathsf{K}$ simply to denote Kashin's representation. 
We call the matrix $U_{\mathsf{K}}$ the Kashin transformation. 
\end{lemma}
We remark that Kashin's representation was first used for locally-private mean estimation in~\cite{FeldmanGV21}.


\section{Multiple Messages}
\label{sec:multi-msg}
In this section, we study algorithms for vector aggregation in the shuffle model of privacy when each user is permitted to send multiple messages. 
In particular, we study the number of messages that each user should send so that the resulting protocol can achieve the same mean squared error as the optimal mechanism in the central setting of DP. 
We first show a lower bound for the number of messages that must be sent per user to achieve the best possible error while guaranteeing DP. 
We then give an algorithm with matching number of messages per user, while achieving the best possible error for DP protocols. 
\subsection{\texorpdfstring{$\tilde \Omega({\min(n \eps^2,d)})$}{Omega(min(eps2n,d))} messages are necessary}
In this section, we prove that any unbiased shuffle DP protocol that obtains optimal error must send at least $k \ge \Omega\left( \frac{\min(n \eps^2,d)}{\log n} \right)$ messages per user. We prove this lower bound for summation protocols in~\Cref{sec:agg-summation} and for any unbiased protocol in~\Cref{sec:agg-unbiased}.

In our setting, we have $n$ users with inputs $v_1,\dots,v_n \in \reals^d$ where $\norm{v_i}_2 \le 1$. Each user applies a local randomizer $\calR(v_i)$ which sends $k$ messages, $\calR(v_i) = (m_i^1,\dots,m_i^k)$, then an aggregation protocol $\calA$ is applied over the shuffled messages, producing an output $\hat \mu = \calA(\Pi(\calR(v_1),\calR(v_2),\dots,\calR(v_n)))$, where $\Pi$ is the shuffling operation.

\subsubsection{Lower bound for summation protocols}
\label{sec:agg-summation}
We begin by proving the lower bound for summation protocols where $\calA^+(m_1,\dots,m_{nk}) = \sum_{i=1}^{nk} m_i$. Throughout this section, we assume that the aggregation protocol $\calA = \calA^+$ and that $\calR: \sphere^{d-1} \to \calZ^k$ where $\calZ = \reals^d$.

\begin{restatable}{theorem}{thmmainlb}
\label{thm:main-lb}
    Let $\eps,\delta \le 1$ and  $\calR: \sphere^{d-1} \to \calZ^k $ be an $(\eps,\delta)$-Shuffle DP randomizer. If $\Err(\calA^+,\calR) \le \O{d/\eps^2}$ then $k \ge \Omega\left( \frac{\min(n \eps^2,d)}{\log n} \right)$.    
\end{restatable}

Towards proving this result, we first prove the following symmetry property that is satisfied by an optimal summation protocol. For a randomizer $\calR$, let $\calR^+(v_i)$ denote the summation of all messages in $\calR(v_i)$, that is, $\calR^+(v_i) = \sum_{j=1}^k \calR(v_i)_j$. We defer the proof to~\Cref{apdx:proof-struc-opt}.
\begin{restatable}{lemma}{lemmastrucopt}
\label{lemma:struc-opt}
    Let $\eps \le 1$, $\calR: \sphere^{d-1} \to \calZ^k $ be  $(\eps,\delta)$-Shuffle DP. There exists an $(\eps,\delta)$-Shuffle DP randomizer $\hat \calR: \sphere^{d-1} \to \calZ^k $ such that 
    \begin{enumerate}
        \item $\Err(\calA^+,\hat \calR) \le \Err(\calA^+,\calR)$
        \item (Symmetry) For all $u, v \in \sphere^{d-1}$, 
        \begin{equation*}
            \Ex{\norm{\hat \calR^+(v) - v}_2^2} = \Ex{\norm{\hat \calR^+(u) - u}_2^2}
        \end{equation*}   
    \end{enumerate}
\end{restatable}
\begin{proof}(sketch)
    The new randomizer $\hat \calR$ works as follows: first, it samples a rotation matrix $ U \in \reals^{d \times d}$ (known public randomness) such that $U^T U = I$, then sets 
    \begin{equation*}
        \hat \calR(v) = U^T \calR(Uv),
    \end{equation*}
    where $ U^T \calR(Uv)$ denotes multiplying each message in $\calR(Uv)$ by $U^T$. The lemma then follows using standard algebraic manipulations (see~\Cref{apdx:proof-struc-opt} for full proof).
\end{proof}

The proof of the lower bound builds on the following reconstruction attack against summation protocols. The attack essentially iterates over all subsets of messages of size $k$ and adds their sum to the output set. We argue that if the protocol has small error (less than $n$), then the input vector will be close to a vector in the output set.
\begin{algorithm}[!htb]
\caption{Reconstruction attack against summation protocols}
\label{alg:atck-sum}
\begin{algorithmic}[1]
\REQUIRE{Shuffled set of messages $W = \{m_i\}_{i \in [nk]} \in \reals^d$}
\ENSURE{A set $S \subseteq \reals^d$ }
\STATE{Initialize $S  = \emptyset $}
\FOR{$t=1$ to $\binom{nk}{k}$}
\STATE{Pick a (new) set of $k$ messages from $W$; denote it by $W_t$}
\STATE{$S \gets S \cup \{\sum_{m \in W_t} m \}$}
\ENDFOR
\STATE{Return $S$}
\end{algorithmic}
\end{algorithm}

The following proposition states the guarantees of this reconstruction attack.
\begin{restatable}{proposition}{proprecattacksum}
\label{prop:rec-attck-sum}
	Let $v_1,\dots,v_n \in \sphere^{d-1}$, $\calR: \sphere^{d-1} \to \calZ^k$ be randomizer that satisfies the symmetry condition of~\Cref{lemma:struc-opt}. 
    For an input set $W = \Pi(R(v_1),\dots,R(v_n))$, \Cref{alg:atck-sum} outputs a set $S \subset \reals^d$ of size $\binom{nk}{k}$ such that  
	\begin{equation*}
		\Ex{\dist(v_1,S)} = \Ex{\min_{u \in S} \ltwo{v_1 - u}^2} \le \frac{\Err(\calA^+,\calR)}{ n },
	\end{equation*}
	where the expectation is over the randomness of the algorithm.
\end{restatable}
\iftoggle{arxiv}{
    \begin{proof}
    Let $\Delta(v) = E \left[\calR^+(v) -  v \right]$ be the bias of $\calR^+$ over $v$. Note that the error of the protocol over dataset $(u_1,\dots,u_n)$ is  
    \begin{align*}
    & E \left[\norm{\calA^+(\Pi(\calR(u_1),\calR(u_2),\dots,\calR(u_n))) - \sum_{i=1}^n u_i}_2^2 \right] \\
        & =  E \left[\norm{\sum_{i=1}^n \calR^+(u_i) -  u_i}_2^2 \right] \\
        & = \sum_{i=1}^n  E \left[\norm{\calR^+(u_i) -  u_i}_2^2 \right] + \sum_{i \neq j \in [n]}  E \left[\calR^+(u_i) -  u_i \right]^T E \left[\calR^+(u_j) -  u_j \right] \\
        & = \sum_{i=1}^n  E \left[\norm{\calR^+(u_i) -  u_i}_2^2 \right] + \sum_{i \neq j \in [n]} \Delta(u_i)^T \Delta(u_j).
    \end{align*}
    For input dataset $X =  (u,u,\dots,u)$, this implies
    \begin{align*}
    & E \left[\norm{\calA^+(\Pi(\calR(u),\calR(u),\dots,\calR(u))) - n u}_2^2 \right] \\
        & = n E \left[\norm{\calR^+(u) - u}_2^2  \right] + \binom{n}{2} \norm{\Delta(u)}_2^2 \\
        & \ge  n E \left[\norm{\calR^+(u) - u}_2^2  \right].
    \end{align*}
    As $\calR$ satisfies the symmetry assumption that $E[\norm{\calR^+(v) - v}_2^2] = E[\norm{\calR^+(u) - u}_2^2] $ for all $u,v \in \sphere^{d-1}$, and since the error is bounded by $d/\eps^2$, we have that 
    \begin{equation*}
    E \left[\norm{\calR^+( v_1) -   v_1}_2^2\right] \le \frac{d}{n \eps^2}.
    \end{equation*}
    Finally, note that $\calR^+(v_1) \in S$ as the attack of~\Cref{alg:atck-sum} iterates over all possible subsets of size $k$ and adds their sum to $S$. Hence, there exists $t$ such that $W_t = \calR(v_1)$, in which case the algorithm will add $\calR^+(v_1)$ to $S$.
\end{proof}

}

We can now provide the main idea for proving~\Cref{thm:main-lb}. We defer the full proof to~\Cref{apdx:proof-main-lb-sum}.
\begin{proof}(sketch)
    Consider $d \le n \eps^2/100 $ and let $P = \{v_1,v_2,\dots,v_M \}$ be a $\rho$-packing of $\sphere^{d-1}$ where $\rho = 1/10$. 
    We will prove the lower bounds by analyzing the algorithm over the following $M$ datasets:
    \begin{equation*}
        X_i = (v_i,v_1,\dots,v_1).
    \end{equation*}
    The main idea is to show that if an algorithm is accurate, then our reconstruction attack~\Cref{alg:atck-sum} will return a set $S$ of size $\binom{nk}{k} \approx 2^{k \log(n)}$ that contains $v_i$. If $k\ll d$, then the size of the reconstructed set $S$ is much smaller than the size of the packing $P$ which contradicts privacy.

    More formally, 
    let $S_i$ be the output of the reconstruction attack (\Cref{alg:atck-sum}) over the input $\Pi(\calR(v_i),\calR(v_1),\dots,\calR(v_1))$, and let $O_i$ be the projection of $S_i$ to the packing $P$; that is, $O_i = \{ \mathsf{Proj}_P(v) : v\in S_i \} $.

    \Cref{prop:rec-attck-sum} states that $\E[\dist(v_i,S_i)] \le \frac{d}{n \eps^2} \le 1/100$, hence we get that 
    \begin{equation*}
    \Pr(v_i \in O_i ) \ge \Pr(\dist(v_i,S_i)<\rho) \ge 9/10,
    \end{equation*}
    where the first inequality follows as $P$ is $\rho$-packing, and the second inequality follows from Markov's inequality.

    On the other hand, note that
    \begin{align*}
    \sum_{i =1}^M \Pr(v_i \in O_1)
        & = \sum_{i =1}^M E[1\{v_i \in O_1\}]  \\
        & = E\left[\sum_{i =1}^M 1\{v_i \in O_1\} \right]  \\
        & \le E[|O_1|] \le \binom{nk}{k}.
    \end{align*}
    Hence there exists an $1 \le i \le M$ such that 
    \begin{equation*}
        \Pr(v_i \in O_1) \le \frac{\binom{nk}{k}}{M}.
    \end{equation*}
    As the protocol is $(\eps,\delta)$-DP, we also have 
    \begin{align*}
    \Pr(v_i \in O_1) 
        & \ge \Pr(v_i \in O_i) e^{-\eps} - \delta\\
        & \ge \frac{9}{10 e} \ge 1/6.
    \end{align*}
    Combining these together, and given that $M \ge 2^d$ for $\rho = 1/10$, we have that 
    \begin{equation*}
        2^d \le 6 \binom{nk}{k} \le 6 (en)^k.
    \end{equation*}
    This implies that $k \ge \Omega(d/\log(n))$ whenever $d \le n\eps^2/100$. 

Now consider $d \ge n \eps^2/100$. The proof builds on a reduction (\Cref{prop:reduc-small-to-large-dim}) which converts an optimal protocol for $d$-dimensional inputs into an optimal protocol for $d'$-dimensional inputs where $ d' = n \eps^2/200$ with the same number of messages. The lower bound then follows immediately from the lower bound for small $d$.

We provide the full details of the proof and the missing proof for $d \ge n\eps^2/100$ in \appref{app:apdx-agg-summ}.
\end{proof}

\subsubsection{Lower bound for unbiased protocols}
\label{sec:agg-unbiased}
In this section, we prove the same lower bound for any aggregation strategy as long as it is unbiased.
We assume that the aggregation protocol $\calA$ is unbiased; that is, for all $v_1,\dots,v_n \in \sphere^{d-1}$, 
\begin{equation*}
\Ex{\calA(\Pi(\calR(v_1),\calR(v_2),\dots,\calR(v_n)))} = \sum_{i=1}^n v_i.
\end{equation*}

The lower bound builds on the following reconstruction attack against unbiased protocols (\Cref{alg:atck-unb}). The attack follows the same recipe as the attack against summation protocols (\Cref{alg:atck-sum}) to iterate over all subsets of messages of size $k$. However, given $k$ messages, now we apply a different reconstruction scheme that uses the aggregation $\calA$ with zero-mean dummy inputs, and finally taking expectations.

\begin{algorithm}[!htb]
\caption{Reconstruction attack against unbiased protocols}
\label{alg:atck-unb}
\begin{algorithmic}[1]
\REQUIRE{Shuffled set of messages $W = \{m_i\}_{i \in [nk]}$}
\ENSURE{A set $S \subseteq \reals^d$ }
\STATE{Initialize $S  = \emptyset $}
\FOR{$t=1$ to $ \binom{nk}{k}$}
\STATE{Pick a (new) set of $k$ messages from $W$; denote it by $W_t$}
\STATE{Calculate $u_t$ to be $\E_{\tilde v_2,\dots, \tilde v_{n} \sim \mathsf{Unif}(\sphere^{d-1})} \left[ \calA\left(\Pi(W_t, \calR(\tilde  v_2),\dots,\calR(\tilde  v_{n})) \right) \right] $}
\STATE{$S \gets S \cup \{u_t \}$}
\ENDFOR
\STATE{Return $S$}
\end{algorithmic}
\end{algorithm}

The following proposition states the guarantees of this reconstruction attack against unbiased protocols. We defer the proof to~\Cref{apdx:proof-rec-attck-unb}.
\begin{restatable}{proposition}{proprecattackunb}
\label{prop:rec-attck-unb}
	Let $v_1,\dots,v_n \in \sphere^{d-1}$ where  $v_1 \sim \mathsf{Unif}(\sphere^{d-1})$, $\calR: \sphere^{d-1} \to \calZ^k$ and $\calA$ be an unbiased protocol. 
    For an input set $W = \Pi(R(v_1),\dots,R(v_n))$, \Cref{alg:atck-unb} outputs a set $S \subset \reals^d$ of size $\binom{nk}{k}$ such that  
	\begin{equation*}
		\Ex{\dist(v_1,S)} = \Ex{\min_{u \in S} \ltwo{v_1 - u}^2} \le \frac{\Err(\calA,\calR)}{ n },
	\end{equation*}
	where the expectation is over the randomness of  $v_1$ and the algorithm.
\end{restatable}

We can now prove our main lower bound for unbiased protocols. The proof is similar to the proof of~\Cref{thm:main-lb} using the new construction attack. We defer it to~\Cref{apdx:main-lb-unb}.
\begin{restatable}{theorem}{thmmainlbunb}
\label{thm:main-lb-unb}
    Let $\eps \le 1$, $\calR: \sphere^{d-1} \to \calZ^k $ be $(\eps,\delta)$-shuffle DP, and $\calA$ be an unbiased protocol. If $\Err(\calA,\calR) \le \O{d/\eps^2}$ then $k \ge \Omega\left( \frac{\min(n \eps^2,d)}{\log n} \right)$. 
\end{restatable}

\subsection{Optimal multi-message protocol}
\label{sec:opt:mult}
In this section, we briefly overview a private protocol that achieves the optimal mean squared error for vector aggregation. 
The protocol requires that each user sends $\O{d}$ messages in expectation. 

\newcommand{\Rg}{\calR_{\mathsf{GKMPS}}} 
We adapt the $1$-dimensional mechanism of \cite{GhaziKMPS21} to vector aggregation by requiring that each user separately performs the scalar aggregation on each coordinate and padding the resulting messages to vectors in the natural manner, before sending the messages. 
Due to standard composition theorems, the privacy parameter for each coordinate must have a smaller privacy budget, so that the overall privacy loss across the $d$ coordinates is still $\eps$. 
We describe the local randomizer in~\Cref{alg:rand-large-d} and the aggregation in~\Cref{alg:agg-large-d}. Our algorithms use the optimal $1$-dimensional algorithm of~\cite{GhaziKMPS21}: we let $\Rg^{(\eps,\delta)}$ denote their local randomizer with parameters $(\eps,\delta)$

\begin{algorithm}[!htb]
\caption{Local randomizer for vector aggregation}
\label{alg:rand-large-d}
\begin{algorithmic}[1]
\REQUIRE{$v^{(i)}\in \sphere^{d-1}$, privacy parameters $(\eps,\delta)$}
\ENSURE{$S^{(i)}\subset \reals^d$}
\STATE{Let $S^{(i)} = \emptyset$ and $U_{\mathsf{K}} \in \reals^{2d \times d}$ be the Kashin transformation with constant $\Ck$}
\STATE{Set $u^{(i)} = \frac{\sqrt{d}}{\Ck}  U_{\mathsf{K}} v^{(i)}$}
\FOR{$j=1$ to $2d$}
\STATE{Let $S_j = \Rg^{(\eps_0,\delta_0)}( u^{(i)}_j)$ where $\eps_0 = \frac{\eps}{2\sqrt{2d \log(2/\delta)}}$ and $\delta_0 = \frac{\delta}{2d}$}
\STATE{Update $S^{(i)} = S^{(i)} \cup \{m \cdot e_j : m \in S_j \}$}
\ENDFOR
\STATE{Output $S^{(i)}$}
\end{algorithmic}
\end{algorithm}

\begin{algorithm}[!htb]
\caption{Aggregation for vector aggregation}
\label{alg:agg-large-d}
\begin{algorithmic}[1]
\REQUIRE{Shuffled messages $M \subset \reals^{2d}$}
\ENSURE{$\hat v \in \reals^d$}
\STATE{Let $U_{\mathsf{K}} \in \reals^{2d \times d}$ be the Kashin transformation with constant $\Ck$}
\STATE{Calculate $\hat u =  \sum_{m \in M} m$}
\STATE{Output $\hat v = \frac{\Ck}{\sqrt{d}}  U^T_{\mathsf{K}} \hat u$}
\end{algorithmic}
\end{algorithm}

We have the following result for our protocol. 
The proof is standard and we defer it to~\Cref{sec:proofs-ub-multi}.
\begin{restatable}{theorem}{thmmultmsgerrbound}
\label{thm:mult:msg:err:bound}
Let $\calR: \sphere^{d-1} \to \reals^{2d}$ be the local randomizer in~\Cref{alg:rand-large-d} and $\calA: (\reals^{2d})^\star \to \reals^d$ be the aggregation in~\Cref{alg:agg-large-d}. Then, $\calR$ is $(\eps,\delta)$-Shuffle DP randomizer, each users sends $d \cdot \left( 1 + \widetilde{\mathcal{O}}_{\eps}\left(\frac{\log(1/\delta)}{\sqrt{n}}\right) \right)$ messages in expectation, and the protocol has error
\begin{equation*}
\Err(\calA,\calR) \le \O{ \frac{d \log(1/\delta)}{\eps^2}}.
\end{equation*}
\end{restatable}

Finally, we note that it is possible to achieve this rate with $\O{n\eps^2}$ messages using the protocols in~\cite{chen2023privacy}: their protocols work in the shuffle model and send $\O{n\eps^2}$ bits per users in $T$ rounds. However, their approach can also work in a single round if the coordinates are independent (which is the case if the Kashin representation is applied). Overall, we conclude that there is a protocol for the shuffle model that requires $\O{\min(n\eps^2,d)}$ messages.

\section{Single Message per User}
\label{sec:single-msg}
In this section, we study private vector summation when each user is only allowed to send a single message. 
We first give an algorithm for this setting in~\Cref{sec:ub-single} and then show that the algorithm is near-optimal in~\Cref{sec:lb-single}.

\subsection{A Single-Message Protocol}
\label{sec:ub-single}
\seclab{sec:single:alg}
In this section, we describe a simple protocol for private vector summation in the shuffle model that achieves near-optimal error when each user can only send a single message and $\varepsilon$ is constant.
Indeed, both the protocol and the corresponding analysis can be viewed as a generalization of \cite{BalleBGN19} from the aggregation of real numbers to real-valued vectors. 

The protocol first picks a granularity $r$ so that all messages will only correspond to vectors whose coordinates are multiples of $r$. 
Each user $i$ then randomly rounds each coordinate of their input $v^{(i)}$ to one of the two neighboring multiples of $r$ to form a vector $\widetilde{v^{(i)}}$. 
Each user then performs randomized response to determine whether the message $w^{(i)}$ they send is their randomly rounded input $\widetilde{v^{(i)}}$ or a message selected uniform at random from the set $[r]^d$ of all possible rounded messages. 
The local randomizer appears in \algref{alg:local:single}. 

The analyzer takes the set $\{w^{(i)}\}_{i\in[n]}$ of messages and computes their vector sum $z=\sum_{i=1}^n w^{(i)}$. 
It then adjusts each coordinate $j\in[d]$ of $z$ to account for the expected noise from randomized response, so that the expectation of the corrected $z_j$ is precisely the sum of the inputs $\sum_{i=1}^n v^{(i)}_j$. 
We provide the full details in \algref{alg:bucket:rr}. 

\begin{algorithm}[!htb]
\caption{Local randomizer for single message per user}
\alglab{alg:local:single}
\begin{algorithmic}[1]
\REQUIRE{$v^{(i)}\in\sphere^{d-1}$, parameters $r, c, d, n$}
\ENSURE{$w^{(i)}\in\{0,1,\ldots,r\}^d$}
\STATE{$\gamma\gets\frac{c(r+1)}{n}$}
\FOR{$j=1$ to $j=d$}
\STATE{$\widetilde{v^{(i)}}_j\sim\left\lfloor rv^{(i)}_j\right\rfloor+\Ber\left(rv^{(i)}_j-\left\lfloor rv^{(i)}_j\right\rfloor\right)$}
\ENDFOR
\STATE{Sample $b\sim\Ber(\gamma)$}
\IF{$b=0$}
\STATE{$w^{(i)}\gets\widetilde{v^{(i)}}$}
\ELSE
\STATE{$w^{(i)}\sim\Unif([r]^d)$}
\ENDIF
\end{algorithmic}
\end{algorithm}

\begin{algorithm}[!htb]
\caption{Aggregation for bucket-based randomized response}
\alglab{alg:bucket:rr}
\begin{algorithmic}[1]
\REQUIRE{$w^{(i)}\in\{0,1,\ldots,r\}^d$ for $i\in[n]$ and $c$ from \algref{alg:local:single}}
\ENSURE{$\widetilde{v}\in[0,n]^d$}
\STATE{$z\gets\frac{1}{r}\sum_{i=1}^n w^{(i)}$}
\FOR{$j=1$ to $j=d$}
\STATE{$\widetilde{v}_j\gets\frac{2z_j-c(r+1)}{2-2\gamma}=\left(z_j-\frac{c(r+1)}{2}\right)/(1-\gamma)$}
\ENDFOR
\STATE{Return $\widetilde{v}$}
\end{algorithmic}
\end{algorithm}
We first note that since each vector in $[r]^d$ can be encoded as a integer in $[r^d]$, then the privacy guarantees of \cite{BalleBGN19} for the local randomizer holds as follows: 
\begin{lemma}[Theorem 3.1 in~\cite{BalleBGN19}]
\lemlab{lem:sdp:amp:bucket}
The mechanism in \algref{alg:bucket:rr} is $(\eps,\delta)$-DP for the number of buckets $k=(r+1)^d$ and $\gamma\ge\min\left(1,\max\left(\frac{14k}{(n-1)\eps^2}\log\frac{2}{\delta},\frac{27k}{(n-1)\eps}\right)\right)$. 
\end{lemma}
We now upper bound the mean squared error of \algref{alg:bucket:rr}. 
\thmonemsgub*
\begin{proof}
Consider \algref{alg:bucket:rr}. 
The mechanism is $(\eps,\delta)$-private by the choice of 
\[\gamma\ge\min\left(1,\max\left(\frac{14k}{(n-1)\eps^2}\log\frac{2}{\delta},\frac{27k}{(n-1)\eps}\right)\right)\]
and \lemref{lem:sdp:amp:bucket}. 

The mean squared error is at most
\begin{align*}
\sup_{\{v^{(i)}\}}\Ex{\sum_{j=1}^d(\widetilde{v}_j-v_j)^2}&=\sup_{\{v^{(i)}\}}\Ex{\sum_{j=1}^d\left(\widetilde{v}_j-\sum_{i=1}^n v^{(i)}_j\right)^2}.
\end{align*}
For a real number $x$, let $F(x)=\frac{x-c(r+1)/2}{1-c(r+1)/n}$, so that $F$ is the debiasing function applied coordinate-wise to $z$.  
\begin{align*}
\sup_{\{v^{(i)}\}}\Ex{\sum_{j=1}^d(\widetilde{v}_j-v_j)^2}&=\sup_{\{v^{(i)}\}}\Ex{\sum_{j=1}^d\left(F(z_j)-\sum_{i=1}^n v^{(i)}_j\right)^2}\\
&=\sup_{\{v^{(i)}\}}\Ex{\sum_{j=1}^d\left(F\left(\frac{1}{r}\sum_{i=1}^nw^{(i)}_j\right)-\sum_{i=1}^n v^{(i)}_j\right)^2}.
\end{align*}
Note that by construction $\Ex{F(z_j)}=\sum_{i=1}^n v^{(i)}_j$, for all $j\in[d]$. 
Thus the cross terms cancel, so that we further have
\begin{align*}
\sup_{\{v^{(i)}\}}\Ex{\sum_{j=1}^d(\widetilde{v}_j-v_j)^2}&=\sup_{\{v^{(i)}\}}\Ex{\sum_{j=1}^d\sum_{i=1}^n\left(F\left(\frac{w^{(i)}_j}{r}\right)-\sum_{i=1}^n v^{(i)}_j\right)^2}\\
&=\sup_{\{v^{(i)}\}}\sum_{j=1}^d\sum_{i=1}^n\Var{F\left(\frac{w^{(i)}_j}{r}\right)},
\end{align*}
where we use $\mathbb{V}$ to denote the variance. 
Note that after debiasing, the $\gamma$ fraction of the coordinates that were randomly generated from the uniform distribution, due to $b\sim\Ber(\gamma)$, do not contribute variance. 
Hence the mean-squared error is at most
\begin{align*}
\sup_{\{v^{(i)}\}}\Ex{\sum_{j=1}^d(\widetilde{v}_j-v_j)^2}&=\frac{nd}{(1-\gamma)^2}\sup_{x^{(1)}_1}\Var{\frac{w^{(1)}_1}{r}}\\
&\le\frac{nd}{(1-\gamma)^2}\left(\frac{1-\gamma}{4r^2}+\frac{\gamma}{2}\right).
\end{align*}
Recall that we set $\gamma=\frac{c(k+1)}{n}$ for a parameter $c$, which we require to guarantee 
\[\gamma\ge\min\left(1,\max\left(\frac{14k}{(n-1)\eps^2}\log\frac{2}{\delta},\frac{27k}{(n-1)\eps}\right)\right),\]
to satisfy privacy. 
Then we have
\begin{align*}
\sup_{\{v^{(i)}\}}\Ex{\sum_{j=1}^d(\widetilde{v}_j-v_j)^2}&\le\frac{nd}{(1-\gamma)^2}\left(\frac{1}{4r^2}+\frac{c(k+1)}{2n}\right)\\
&\le\frac{nd}{(1-\gamma)^2}\left(\frac{1}{4r^2}+\frac{c((r+1)^d+1)}{2n}\right).
\end{align*}
By setting $c(r+1)^{d+2}=\O{n}$ for $c=\O{\frac{1}{\eps^2}\log\frac{1}{\delta}}$, we have that the above quantity is minimized at $r=\left(\mathcal{O}_{\delta}\left(\frac{n}{c}\right)\right)^{1/(d+2)}$. 
Thus since $c=\O{\frac{1}{\eps^2}\log\frac{1}{\delta}}$, then the mean squared error is at most $\mathcal{O}_{\delta}\left(dn^{d/(d+2)}\eps^{-4/(d+2)}\right)$.
\end{proof}

\subsection{Lower Bound}
\label{sec:lb-single}
In this section, we show that our protocol in \secref{sec:single:alg} is near-optimal by proving that for any $\eps=\O{1}$, the mean squared error of any protocol that gives $\eps$-DP in the shuffle model in which each user sends a single message is $\Omega(dn^{d/(d+2)})$. 
The main intuition is that we can partition the space into blocks of size length $\frac{1}{r}$, so that there are $r^d$ hypercubes in total. 
Although $r$ is a parameter that can be chosen at the protocol's discretion, there are two sources of error for any private protocol that result in two opposing tensions on the value of $r$. 

The first source of error is that due to the privacy guarantees, the output distribution for an input $v^{(i)}$ to a player $i$ may overlap with the output distribution for any input in $[r]^d$. 
In this case, the message may be decoded to some other vector with large distance from $v^{(i)}$, resulting in large mean squared error. 
In particular, larger values of $r$ force the output of the local randomizer to have less signal about the true block containing the input $v^{(i)}$, since the output distribution must intersect with that of more possible inputs. 
This is formalized in \lemref{lem:atob}. 

The second source of error is that any vector inside a block may incur error from the message representing the block, due to the partition of the space. 
In particular, the message may be decoded correctly for the block, but the set of all vectors within the block has large diameter, and so the resulting mean squared error is large. 
Specifically, smaller values of $r$ result in blocks with larger diameter, which again force the output of the local randomizer to have less signal about the true input $v^{(i)}$ within each block. 
This is formalized in \lemref{lem:err:inblock}. 
The resulting lower bound then follows from optimizing $r$ with respect to the two possible sources of error, resulting in the following theorem.
\thmonemsglb*

The proof of this result is technical and we defer it to \appref{app:single:msg}.

\section{Robustness to Malicious Users}
\label{sec:rob}
We first observe that our multi-message protocol is not robust against malicious users in the single-shuffle setting, in the sense that a single malicious user can additively incur much larger than constant mean squared error, even though their input vector has at most unit length. 
In fact, each user can incur up to $\Omega\left(\frac{k}{\log^2(nd)}\right)$ additive mean squared error. 
\begin{restatable}{theorem}{thmattackoneshuffler}
\thmlab{thm:attack:one:shuffler}
Let $\eps=\O{1}$ and $\delta<\frac{1}{nd}$. 
Then any $(\eps,\delta)$-shuffle DP mechanism for vector summation that takes the sum of the messages across $n$ players with $k$ malicious users has additive error $\Omega\left(\frac{kd}{\log^2(nd)}\right)$. 
\end{restatable}
\thmref{thm:bad:users:err} then follows from a simple power mean inequality. 
On the other hand, we observe that~\Cref{alg:agg-large-d} is robust against malicious users in the setting where a separate shuffler is responsible for the messages corresponding to each coordinate of a user.
\begin{lemma}
Suppose that a separate shuffler handles the messages for each coordinate from all users in~\Cref{alg:agg-large-d}. 
Then the mean squared error induced by $k$ malicious users is at most $\O{k}$. 
\end{lemma}


\bibliography{references}

\newcommand{\etalchar}[1]{$^{#1}$}
\begin{thebibliography}{MMR{\etalchar{+}}17}

\bibitem[ACG{\etalchar{+}}16]{AbadiCGMMT016}
Mart{\'{\i}}n Abadi, Andy Chu, Ian~J. Goodfellow, H.~Brendan McMahan, Ilya
  Mironov, Kunal Talwar, and Li~Zhang.
\newblock Deep learning with differential privacy.
\newblock In {\em Proceedings of the 2016 {ACM} {SIGSAC} Conference on Computer
  and Communications Security}, pages 308--318, 2016.

\bibitem[AFN{\etalchar{+}}23]{AsiFeNeNgTa23}
Hilal Asi, Vitaly Feldman, Jelani Nelson, Huy Nguyen, and Kunal Talwar.
\newblock Fast optimal locally private mean estimation via random projections.
\newblock In {\em Thirty-seventh Conference on Neural Information Processing
  Systems}, 2023.

\bibitem[AFT22]{AsiFeTa22}
Hilal Asi, Vitaly Feldman, and Kunal Talwar.
\newblock Optimal algorithms for mean estimation under local differential
  privacy.
\newblock In {\em International Conference on Machine Learning, {ICML}, {USA}},
  pages 1046--1056, 2022.

\bibitem[AS19]{AcharyaS19}
Jayadev Acharya and Ziteng Sun.
\newblock Communication complexity in locally private distribution estimation
  and heavy hitters.
\newblock In {\em Proceedings of the 36th International Conference on Machine
  Learning, {ICML}}, pages 51--60, 2019.

\bibitem[BBGN19]{BalleBGN19}
Borja Balle, James Bell, Adri{\`{a}} Gasc{\'{o}}n, and Kobbi Nissim.
\newblock The privacy blanket of the shuffle model.
\newblock In {\em Advances in Cryptology - {CRYPTO} 2019 - 39th Annual
  International Cryptology Conference, Proceedings, Part {II}}, pages 638--667,
  2019.

\bibitem[BBGN20]{BalleBGN20}
Borja Balle, James Bell, Adri{\`{a}} Gasc{\'{o}}n, and Kobbi Nissim.
\newblock Private summation in the multi-message shuffle model.
\newblock In {\em {CCS} '20: 2020 {ACM} {SIGSAC} Conference on Computer and
  Communications Security}, pages 657--676, 2020.

\bibitem[BDF{\etalchar{+}}18]{BhowmickDFKR18}
Abhishek Bhowmick, John~C. Duchi, Julien Freudiger, Gaurav Kapoor, and Ryan
  Rogers.
\newblock Protection against reconstruction and its applications in private
  federated learning.
\newblock {\em CoRR}, abs/1812.00984, 2018.

\bibitem[BEM{\etalchar{+}}17]{BittauEMMRLRKTS17}
Andrea Bittau, {\'{U}}lfar Erlingsson, Petros Maniatis, Ilya Mironov, Ananth
  Raghunathan, David Lie, Mitch Rudominer, Ushasree Kode, Julien Tinn{\'{e}}s,
  and Bernhard Seefeld.
\newblock Prochlo: Strong privacy for analytics in the crowd.
\newblock In {\em Proceedings of the 26th Symposium on Operating Systems
  Principles}, pages 441--459, 2017.

\bibitem[BHNS20]{beimel2020round}
Amos Beimel, Iftach Haitner, Kobbi Nissim, and Uri Stemmer.
\newblock On the round complexity of the shuffle model.
\newblock In {\em Theory of Cryptography Conference}, pages 683--712. Springer,
  2020.

\bibitem[BNO08]{BeimelNO08}
Amos Beimel, Kobbi Nissim, and Eran Omri.
\newblock Distributed private data analysis: Simultaneously solving how and
  what.
\newblock In {\em Advances in Cryptology - {CRYPTO} 2008, 28th Annual
  International Cryptology Conference. Proceedings}, pages 451--468, 2008.

\bibitem[BS15]{BassilyS15}
Raef Bassily and Adam~D. Smith.
\newblock Local, private, efficient protocols for succinct histograms.
\newblock In {\em Proceedings of the Forty-Seventh Annual {ACM} on Symposium on
  Theory of Computing, {STOC}}, pages 127--135, 2015.

\bibitem[BST14]{BassilyST14}
Raef Bassily, Adam~D. Smith, and Abhradeep Thakurta.
\newblock Private empirical risk minimization: Efficient algorithms and tight
  error bounds.
\newblock In {\em 55th {IEEE} Annual Symposium on Foundations of Computer
  Science, {FOCS}}, pages 464--473, 2014.

\bibitem[CB17]{Corrigan-GibbsB17}
Henry Corrigan{-}Gibbs and Dan Boneh.
\newblock Prio: Private, robust, and scalable computation of aggregate
  statistics.
\newblock In Aditya Akella and Jon Howell, editors, {\em 14th {USENIX}
  Symposium on Networked Systems Design and Implementation, {NSDI} 2017,
  Boston, MA, USA, March 27-29, 2017}, pages 259--282. {USENIX} Association,
  2017.

\bibitem[CGKM20]{chen2020distributed}
Lijie Chen, Badih Ghazi, Ravi Kumar, and Pasin Manurangsi.
\newblock On distributed differential privacy and counting distinct elements.
\newblock {\em arXiv:2009.09604 [cs.CR]}, 2020.

\bibitem[CJMP22]{CheuJMP22}
Albert Cheu, Matthew Joseph, Jieming Mao, and Binghui Peng.
\newblock Shuffle private stochastic convex optimization.
\newblock In {\em The Tenth International Conference on Learning
  Representations, {ICLR}}, 2022.

\bibitem[CKO20]{ChenKO20}
Wei{-}Ning Chen, Peter Kairouz, and Ayfer \"{O}zg\"{u}r.
\newblock Breaking the communication-privacy-accuracy trilemma.
\newblock In {\em Proceedings of the 33rd Annual Conference on Advances in
  Neural Information Processing Systems (NeurIPS)}, 2020.

\bibitem[CMS11]{ChaudhuriMS11}
Kamalika Chaudhuri, Claire Monteleoni, and Anand~D. Sarwate.
\newblock Differentially private empirical risk minimization.
\newblock {\em J. Mach. Learn. Res.}, 12:1069--1109, 2011.

\bibitem[CSOK23]{chen2023privacy}
Wei-Ning Chen, Dan Song, Ayfer Ozgur, and Peter Kairouz.
\newblock Privacy amplification via compression: Achieving the optimal
  privacy-accuracy-communication trade-off in distributed mean estimation.
\newblock {\em arXiv:2304.01541 [stat.ML]}, 2023.

\bibitem[CSS12]{ChanSS12}
T.{-}H.~Hubert Chan, Elaine Shi, and Dawn Song.
\newblock Optimal lower bound for differentially private multi-party
  aggregation.
\newblock In {\em Algorithms - {ESA} 2012 - 20th Annual European Symposium.
  Proceedings}, pages 277--288, 2012.

\bibitem[CSU{\etalchar{+}}19]{CheuSUZZ19}
Albert Cheu, Adam~D. Smith, Jonathan~R. Ullman, David Zeber, and Maxim
  Zhilyaev.
\newblock Distributed differential privacy via shuffling.
\newblock In {\em Advances in Cryptology - {EUROCRYPT} 2019 - 38th Annual
  International Conference on the Theory and Applications of Cryptographic
  Techniques, Proceedings, Part {I}}, pages 375--403, 2019.

\bibitem[CSU21]{CheuSU19}
Albert Cheu, Adam Smith, and Jonathan Ullman.
\newblock Manipulation attacks in local differential privacy.
\newblock {\em Journal of Privacy and Confidentiality}, 11(1), Feb. 2021.

\bibitem[CU21]{cheu2021limits}
Albert Cheu and Jonathan Ullman.
\newblock The limits of pan privacy and shuffle privacy for learning and
  estimation.
\newblock In {\em Proceedings of the 53rd Annual ACM SIGACT Symposium on Theory
  of Computing}, pages 1081--1094, 2021.

\bibitem[DGS{\etalchar{+}}18]{DavidsonGSTV18}
Alex Davidson, Ian Goldberg, Nick Sullivan, George Tankersley, and Filippo
  Valsorda.
\newblock Privacy pass: Bypassing internet challenges anonymously.
\newblock {\em Proc. Priv. Enhancing Technol.}, 2018(3):164--180, 2018.

\bibitem[DMNS06]{DworkMNS06}
Cynthia Dwork, Frank McSherry, Kobbi Nissim, and Adam~D. Smith.
\newblock Calibrating noise to sensitivity in private data analysis.
\newblock In {\em Theory of Cryptography, Third Theory of Cryptography
  Conference, {TCC}, Proceedings}, pages 265--284, 2006.

\bibitem[DMNS16]{DworkMNS16}
Cynthia Dwork, Frank McSherry, Kobbi Nissim, and Adam~D. Smith.
\newblock Calibrating noise to sensitivity in private data analysis.
\newblock {\em J. Priv. Confidentiality}, 7(3):17--51, 2016.

\bibitem[DR14]{DworkR14}
Cynthia Dwork and Aaron Roth.
\newblock The algorithmic foundations of differential privacy.
\newblock {\em Found. Trends Theor. Comput. Sci.}, 9(3-4):211--407, 2014.

\bibitem[DR19]{DuchiR19}
John~C. Duchi and Ryan Rogers.
\newblock Lower bounds for locally private estimation via communication
  complexity.
\newblock In {\em Conference on Learning Theory, {COLT}}, pages 1161--1191,
  2019.

\bibitem[Duc18]{Duchi18}
John~C. Duchi.
\newblock Introductory lectures on stochastic convex optimization.
\newblock In {\em The Mathematics of Data}, IAS/Park City Mathematics Series.
  American Mathematical Society, 2018.

\bibitem[DWJ16]{DuchiWJ16}
John~C. Duchi, Martin~J. Wainwright, and Michael~I. Jordan.
\newblock Minimax optimal procedures for locally private estimation.
\newblock {\em CoRR}, abs/1604.02390, 2016.

\bibitem[EFM{\etalchar{+}}19]{ErlingssonFeMiRaTaTh19}
Ulfar Erlingsson, Vitaly Feldman, Ilya Mironov, Ananth Raghunathan, Kunal
  Talwar, and Abhradeep Thakurta.
\newblock Amplification by shuffling: From local to central differential
  privacy via anonymity.
\newblock In {\em Proceedings of the Thirtieth ACM-SIAM Symposium on Discrete
  Algorithms (SODA)}, 2019.

\bibitem[FGV21]{FeldmanGV21}
Vitaly Feldman, Crist{\'{o}}bal Guzm{\'{a}}n, and Santosh~Srinivas Vempala.
\newblock Statistical query algorithms for mean vector estimation and
  stochastic convex optimization.
\newblock {\em Math. Oper. Res.}, 46(3):912--945, 2021.

\bibitem[FPE16]{FantiPE16}
Giulia Fanti, Vasyl Pihur, and {\'{U}}lfar Erlingsson.
\newblock Building a {RAPPOR} with the unknown: Privacy-preserving learning of
  associations and data dictionaries.
\newblock {\em Proc. Priv. Enhancing Technol.}, 2016(3):41--61, 2016.

\bibitem[FT21]{FeldmanTa21}
Vitaly Feldman and Kunal Talwar.
\newblock Lossless compression of efficient private local randomizers.
\newblock In {\em Proceedings of the 38th International Conference on Machine
  Learning}, volume 139, pages 3208--3219. PMLR, 2021.

\bibitem[GKM{\etalchar{+}}21]{GhaziKMPS21}
Badih Ghazi, Ravi Kumar, Pasin Manurangsi, Rasmus Pagh, and Amer Sinha.
\newblock Differentially private aggregation in the shuffle model: Almost
  central accuracy in almost a single message.
\newblock In {\em Proceedings of the 38th International Conference on Machine
  Learning, {ICML}}, pages 3692--3701, 2021.

\bibitem[GMPV20]{GhaziMPV20}
Badih Ghazi, Pasin Manurangsi, Rasmus Pagh, and Ameya Velingker.
\newblock Private aggregation from fewer anonymous messages.
\newblock In {\em Advances in Cryptology - {EUROCRYPT} 2020 - 39th Annual
  International Conference on the Theory and Applications of Cryptographic
  Techniques, Proceedings, Part {II}}, pages 798--827, 2020.

\bibitem[HIP{\etalchar{+}}23]{ietf-privacypass-rate-limit-tokens-01}
Scott Hendrickson, Jana Iyengar, Tommy Pauly, Steven Valdez, and Christopher~A.
  Wood.
\newblock {Rate-Limited Token Issuance Protocol}.
\newblock Internet-Draft draft-ietf-privacypass-rate-limit-tokens-01, Internet
  Engineering Task Force, March 2023.
\newblock Work in Progress.

\bibitem[IKOS06]{IshaiKOS06}
Yuval Ishai, Eyal Kushilevitz, Rafail Ostrovsky, and Amit Sahai.
\newblock Cryptography from anonymity.
\newblock In {\em 47th Annual {IEEE} Symposium on Foundations of Computer
  Science {(FOCS} 2006), USA, Proceedings}, pages 239--248. {IEEE} Computer
  Society, 2006.

\bibitem[KLN{\etalchar{+}}11]{KasiviswanathanLNRS11}
Shiva~Prasad Kasiviswanathan, Homin~K. Lee, Kobbi Nissim, Sofya Raskhodnikova,
  and Adam~D. Smith.
\newblock What can we learn privately?
\newblock {\em {SIAM} J. Comput.}, 40(3):793--826, 2011.

\bibitem[LV10]{LyubarskiiV10}
Yurii Lyubarskii and Roman Vershynin.
\newblock Uncertainty principles and vector quantization.
\newblock {\em {IEEE} Trans. Inf. Theory}, 56(7):3491--3501, 2010.

\bibitem[MMR{\etalchar{+}}17]{McMahanMRHA17}
Brendan McMahan, Eider Moore, Daniel Ramage, Seth Hampson, and
  Blaise~Ag{\"{u}}era y~Arcas.
\newblock Communication-efficient learning of deep networks from decentralized
  data.
\newblock In {\em Proceedings of the 20th International Conference on
  Artificial Intelligence and Statistics, {AISTATS}}, pages 1273--1282, 2017.

\bibitem[NXY{\etalchar{+}}16]{NguyenXYHSS16}
Th{\^{o}}ng~T. Nguy{\^{e}}n, Xiaokui Xiao, Yin Yang, Siu~Cheung Hui, Hyejin
  Shin, and Junbum Shin.
\newblock Collecting and analyzing data from smart device users with local
  differential privacy.
\newblock {\em CoRR}, abs/1606.05053, 2016.

\bibitem[ROCT23]{RothblumOCT23}
Guy~N. Rothblum, Eran Omri, Junye Chen, and Kunal Talwar.
\newblock Pine: Efficient norm-bound verification for secret-shared vectors,
  2023.

\bibitem[SCM21]{ScottCM21}
Mary Scott, Graham Cormode, and Carsten Maple.
\newblock Applying the shuffle model of differential privacy to vector
  aggregation.
\newblock In Holger Pirk and Thomas Heinis, editors, {\em Proceedings of the
  The British International Conference on Databases}, volume 3163 of {\em
  {CEUR} Workshop Proceedings}, pages 50--59, 2021.

\bibitem[SCM22]{ScottCM22}
Mary Scott, Graham Cormode, and Carsten Maple.
\newblock Aggregation and transformation of vector-valued messages in the
  shuffle model of differential privacy.
\newblock {\em {IEEE} Trans. Inf. Forensics Secur.}, 17:612--627, 2022.

\bibitem[SFZ{\etalchar{+}}14]{sun2014personalized}
Chongjing Sun, Yan Fu, Junlin Zhou, Hui Gao, et~al.
\newblock Personalized privacy-preserving frequent itemset mining using
  randomized response.
\newblock {\em The Scientific World Journal}, 2014, 2014.

\bibitem[SS15]{ShokriS15}
Reza Shokri and Vitaly Shmatikov.
\newblock Privacy-preserving deep learning.
\newblock In {\em Proceedings of the 22nd {ACM} {SIGSAC} Conference on Computer
  and Communications Security}, pages 1310--1321. {ACM}, 2015.

\bibitem[Tal22]{Talwar22}
Kunal Talwar.
\newblock Differential secrecy for distributed data and applications to robust
  differentially secure vector summation.
\newblock In L.~Elisa Celis, editor, {\em 3rd Symposium on Foundations of
  Responsible Computing, {FORC} 2022, June 6-8, 2022, Cambridge, MA, {USA}},
  volume 218 of {\em LIPIcs}, pages 7:1--7:16. Schloss Dagstuhl -
  Leibniz-Zentrum f{\"{u}}r Informatik, 2022.

\bibitem[TW23]{ietf-ohttp}
Martin Thomson and Christopher~A. Wood.
\newblock {Oblivious HTTP}.
\newblock Internet-Draft draft-ietf-ohai-ohttp-08, Internet Engineering Task
  Force, March 2023.
\newblock Work in Progress.

\bibitem[YB18]{YeB18}
Min Ye and Alexander Barg.
\newblock Optimal schemes for discrete distribution estimation under locally
  differential privacy.
\newblock {\em {IEEE} Trans. Inf. Theory}, 64(8):5662--5676, 2018.

\bibitem[ZWC{\etalchar{+}}22]{Zhou0CFS22}
Mingxun Zhou, Tianhao Wang, T.{-}H.~Hubert Chan, Giulia Fanti, and Elaine Shi.
\newblock Locally differentially private sparse vector aggregation.
\newblock In {\em 43rd {IEEE} Symposium on Security and Privacy, {SP}}, pages
  422--439, 2022.

\end{thebibliography}
\bibliographystyle{alpha}

\appendix

\section{Missing Proofs from \texorpdfstring{\Cref{sec:multi-msg}}{Section 2}}

\subsection{Missing Proofs from \texorpdfstring{\Cref{sec:agg-summation}}{Section 2.1.1}}
\applab{app:apdx-agg-summ}
In this section, we provide the missing proof for the lower bound for the setting of summation protocols (\Cref{sec:agg-summation}).
\subsubsection{Proof of \texorpdfstring{\Cref{lemma:struc-opt}}{Lemma 2.2}}
\label{apdx:proof-struc-opt}
\lemmastrucopt*
\begin{proof}
The new randomizer $\hat \calR$ works as follows: first, it samples a rotation matrix $ U \in \reals^{d \times d}$ (known public randomness) such that $U^T U = I$, then sets 
\begin{equation*}
\hat \calR(v) = U^T \calR(Uv),
\end{equation*}
where $ U^T \calR(Uv)$ denotes multiplying each message in $\calR(Uv)$ by $U^T$.

    To prove privacy, we have to prove that $\Pi(U^T  \calR(Uv_1), U^T \calR(Uv_2),\dots, U^T  \calR( U v_n))$ is $(\eps,\delta)$-DP. As $U$ is known, it is sufficient to prove that $\Pi(\calR(Uv_1), \calR(Uv_2),\dots, \calR( U v_n))$ is $(\eps,\delta)$-DP. This follows directly from the fact that $\Pi(\calR(v_1),\calR(v_2),\dots,\calR(v_n))$ is $(\eps,\delta)$-DP, and that the hamming distance between $X = (v_1,\dots,v_n)$ and $X' = (v'_1,\dots,v'_n)$ is the same as the hamming distance between $X_U = (Uv_1,\dots,Uv_n)$ and $X'_U = (Uv'_1,\dots,Uv'_n)$.  

    \noindent
    For utility, we have
    \begin{align*}
    \Err(\calA^+,\hat \calR)
        & = \sup_{v_1,\dots,v_n} \Ex{\norm{\calA^+(\Pi(\hat \calR(v_1),\hat \calR(v_2),\dots,\hat \calR(v_n))) - \sum_{i=1}^n v_i}_2^2} \\
        & = \sup_{v_1,\dots,v_n} \Ex{\norm{\sum_{i=1}^n \hat \calR^+(v_i) -  v_i}_2^2 } \\
        & = \sup_{v_1,\dots,v_n} \Ex{\norm{\sum_{i=1}^n U^T \calR(Uv_i)  - v_i}_2^2 } \\
        & = \sup_{v_1,\dots,v_n} \Ex{\norm{ U^T \sum_{i=1}^n (\calR(Uv_i)  - Uv_i)}_2^2 } \\
        & = \sup_{v_1,\dots,v_n} \Ex{\norm{ \sum_{i=1}^n (\calR(Uv_i)  - Uv_i)}_2^2 } \\
        & = \Err(A,R).
    \end{align*}

    For the third claim, note that $\hat \calR(-v) = U^T \calR(-Uv)$. As $U$ and $-U$ has the same distribution, we can also write  $\hat \calR(-v) = -U^T \calR(Uv)$ which is the same as the distribution of $-\hat \calR(v)$.

    For the final claim, note that 
    \begin{align*}
    \Ex{\norm{\hat \calR^+(v) - v}_2^2}
        & = \Ex{\norm{U^T \calR(Uv) - v}_2^2} \\
        & = \Ex{\norm{U^T (\calR(Uv) - Uv)}_2^2} \\
        & = \Ex{\norm{ (\calR(Uv) - Uv)}_2^2}.
    \end{align*}
    The claim follows as $Uv_1$ and $U v_2$ have the same distribution for any $v_1$ and $v_2$ in the unit ball.

\end{proof}

\iftoggle{arxiv}{
}
{
\proprecattacksum*
\begin{proof}
    Let $\Delta(v) = E \left[\calR^+(v) -  v \right]$ be the bias of $\calR^+$ over $v$. Note that the error of the protocol over dataset $(u_1,\dots,u_n)$ is  
    \begin{align*}
    & E \left[\norm{\calA^+(\Pi(\calR(u_1),\calR(u_2),\dots,\calR(u_n))) - \sum_{i=1}^n u_i}_2^2 \right] \\
        & =  E \left[\norm{\sum_{i=1}^n \calR^+(u_i) -  u_i}_2^2 \right] \\
        & = \sum_{i=1}^n  E \left[\norm{\calR^+(u_i) -  u_i}_2^2 \right] + \sum_{i \neq j \in [n]}  E \left[\calR^+(u_i) -  u_i \right]^T E \left[\calR^+(u_j) -  u_j \right] \\
        & = \sum_{i=1}^n  E \left[\norm{\calR^+(u_i) -  u_i}_2^2 \right] + \sum_{i \neq j \in [n]} \Delta(u_i)^T \Delta(u_j).
    \end{align*}
    For input dataset $X =  (u,u,\dots,u)$, this implies
    \begin{align*}
    & E \left[\norm{\calA^+(\Pi(\calR(u),\calR(u),\dots,\calR(u))) - n u}_2^2 \right] \\
        & = n E \left[\norm{\calR^+(u) - u}_2^2  \right] + \binom{n}{2} \norm{\Delta(u)}_2^2 \\
        & \ge  n E \left[\norm{\calR^+(u) - u}_2^2  \right].
    \end{align*}
    As $\calR$ satisfies the symmetry assumption that $E[\norm{\calR^+(v) - v}_2^2] = E[\norm{\calR^+(u) - u}_2^2] $ for all $u,v \in \sphere^{d-1}$, and since the error is bounded by $d/\eps^2$, we have that 
    \begin{equation*}
    E \left[\norm{\calR^+( v_1) -   v_1}_2^2\right] \le \frac{d}{n \eps^2}.
    \end{equation*}
    Finally, note that $\calR^+(v_1) \in S$ as the attack of~\Cref{alg:atck-sum} iterates over all possible subsets of size $k$ and adds their sum to $S$. Hence, there exists $t$ such that $W_t = \calR(v_1)$, in which case the algorithm will add $\calR^+(v_1)$ to $S$.
\end{proof}

Given the previous attack, we are now ready to prove our lower bound.
}
\subsubsection{Proof of \texorpdfstring{\Cref{thm:main-lb}}{Theorem 2.1}}
\label{apdx:proof-main-lb-sum}
\thmmainlb*
\begin{proof}
    Let $\Err(\calA,\calR) \le C \cdot d/\eps^2$ for some universal constant $1 \le C < \infty$.
    Based on~\Cref{lemma:struc-opt}, we can assume that the randomizer $R$ satisfies the symmetry property:
    \begin{equation*}
        \Ex{\norm{\calR^+(v) - v}_2^2} = \Ex{\norm{R^+(u) - u}_2^2}, \text{  \quad for all $u,v \in \sphere^{d-1}$.} 
    \end{equation*}

    First, we prove the lower bounds for $d \le n \eps^2/100 C $. 
    Let $P = \{v_1,v_2,\dots,v_M \}$ be a $\rho$-packing of the unit ball such that $M = 2^{d \log(1/\rho)}$ (the existence of such packing is standard in the literature~\cite{Duchi18}). 
    We will prove the lower bounds by analyzing the algorithm over the following $M$ datasets:
    \begin{equation*}
        X_i = (v_i,v_1,\dots,v_1).
    \end{equation*}
    

    Let $S_i$ be the output of the reconstruction attack (\Cref{alg:atck-sum}) over the input $\Pi(\calR(v_i),\calR(v_1),\dots,\calR(v_1))$, and let $O_i$ be the projection of $S_i$ to the packing $P$; that is, $O_i = \{ \mathsf{Proj}_P(v) : v\in S_i \} $.

    \Cref{prop:rec-attck-sum} states that $\Ex{\dist(v_i,S_i)} \le \frac{Cd}{n \eps^2} \le 1/100$, hence we get that 
    \begin{equation*}
    \PPr{v_i \in O_i} \ge \PPr{\dist(v_i,S_i)<\rho} \ge 9/10,
    \end{equation*}
    where the first inequality follows as $P$ is $\rho$-packing, and the second inequality follows from markov inequality.

    
    On the other hand, note that
    \begin{align*}
    \sum_{i =1}^M \PPr{v_i \in O_1}
        & = \sum_{i =1}^M \Ex{1\{v_i \in O_1\}}  \\
        & = \Ex{\sum_{i =1}^M 1\{v_i \in O_1\}}  \\
        & \le \Ex{|O_i|} \le \binom{nk}{k}.
    \end{align*}
    Hence there exists an $1 \le i \le M$ such that 
    \begin{equation*}
        \PPr{v_i \in O_1} \le \frac{\binom{nk}{k}}{M}.
    \end{equation*}
    As the protocol is $(\eps,\delta)$-DP, we also have 
    \begin{align*}
    \PPr{v_i \in O_1} 
        & \ge \PPr{v_i \in O_i} e^{-\eps} - \delta\\
        & \ge \frac{9}{10 e} \ge 1/6.
    \end{align*}
    Combining these together, and given that $M \ge 2^d$ for $\rho = 1/10$, we have that 
    \begin{equation*}
        2^d \le 6 \binom{nk}{k} \le 6 (en)^k.
    \end{equation*}
    This implies that $k \ge \Omega(d/\log(n))$ whenever $d \le n\eps^2/100C$.

    Now we prove the lower bound for $d \ge n \eps^2/100$. The proof builds on the following proposition which states that we can convert an optimal protocol for $d$-dimensional inputs into an optimal protocol for $d'$-dimensional inputs where $ d' = n \eps^2/200$ with the same number of messages.
    We defer the proof to~\Cref{sec:proof-prop-red-small-to-large-dim}.
    \begin{proposition}
    \label{prop:reduc-small-to-large-dim}
    Let $ d' = n \eps^2/200 C \ge 1 $ and $d \ge 2d'$.
    Let $\calR : \sphere^{d-1} \to \calZ^k$ be an $(\eps,\delta)$-shuffle DP protocol with error $\Err(\calA^+,\calR) \le \O{d/\eps^2}$. There exists $\calR' : \sphere^{d'-1} \to \calZ^k$  that is $(\eps,\delta)$-shuffle DP such that $\Err(\calA^+,\calR') \le  \O{d'/\eps^2} $.
\end{proposition}
    
    Now, let $\calA^+$ and $\calR : \ball^{d-1} \to Z^k$ be a protocol that obtains error $\Err(\calA^+, \calR) \le \O{d/\eps^2}$ using $k$ messages. \Cref{prop:reduc-small-to-large-dim} implies that there is a randomizer $\calR': \ball^{d'-1} \to Z^k$ such that $\Err(\calA,\calR') \le\O{d'/\eps^2}$ for $d' = n \eps^2/200C$. As $d' \le n\eps^2/100C$, this shows that $k \ge \Omega(d'/\log(n)) =  \Omega(n \eps^2/\log(n))$.
\end{proof}

\subsubsection{Proof of \texorpdfstring{\Cref{prop:reduc-small-to-large-dim}}{Proposition A.1}}
\label{sec:proof-prop-red-small-to-large-dim}
To prove~\Cref{prop:reduc-small-to-large-dim}, we need the following lemma which shows that we can convert any summation protocol into another one where the error is split evenly across coordinates.

We use the following notation. For a permutation $\pi: [d] \to [d]$ and a vector $v \in \reals^d$, we let $\hat v = v(\pi)$ denote the shuffling of the coordinates of $v$ based on $\pi$, that is $\hat v_j = v_{\pi(j)}$.
\begin{lemma}
\label{lemma:sym-err-rand}
    If $\calR: \sphere^{d-1} \to \calZ^k$ is $(\eps,\delta)$-shuffle DP 
 then there exists $\hat \calR : \{\frac{-1}{\sqrt{d}},\frac{1}{\sqrt{d}} \}^d \to \calZ^k $ that is $(\eps,\delta)$-shuffle DP  and  for $j \in [d]$ and $v_1,\dots,v_n \in \{ \frac{-1}{\sqrt{d}}, \frac{+1}{\sqrt{d}}\}^d$, 
    \begin{equation*}
          \Ex{\left|\left(\sum_{i=1}^n \hat \calR^+(v_i) - \sum_{i=1}^n v_{i} \right)_j\right|^2} \le \frac{\Err(\calA^+,\calR)}{d}.
    \end{equation*}
\end{lemma}

\begin{proof}
    $\hat \calR$ will use shared public randomness to shuffle the coordinates of each vector and flip the signs of each coordinate. This will ensure that all coordinates will have the same marginal distribution for their error. 

    More precisely, let $\pi: [d] \to [d] $ be a random permutation of the coordinates picked uniformly at random, and let $s_1,\dots,s_d \sim \mathsf{Ber}(1/2)$. Our new randomizer $\hat R$ over input $v$ will first transform the input vector $v$ into $\hat v$ where 
    \begin{equation*}
        \hat v = s \cdot v(\pi),
    \end{equation*}
    where the multiplication is element-wise.
    Then, we run $\calR(\hat v)$ to get messages $\hat m_1,\dots,\hat m_k$. For each of these messages, we apply the inverse transformation, and output $ m_1,\dots, m_k$ where 
    \begin{equation*}
        m_{i} = s \cdot \hat m_{i}(\pi^{-1}).
    \end{equation*}

    Note that 
    \begin{align*}
    \left|\left(\sum_{i=1}^n \hat \calR^+(v_i) - \sum_{i=1}^n v_{i} \right)_j \right|^2
        & = \left|\sum_{i=1}^n m_{i,j} -  v_j \right|^2 \\
        & = \left|s_j \sum_{i=1}^n \hat m_{i,\pi^{-1}(j)} - \hat v_{\pi^{-1}(j)} \right|^2 \\
        & = \left| \sum_{i=1}^n \hat m_{i,\pi^{-1}(j)} - \hat v_{\pi^{-1}(j)} \right|^2 \\
        & = \left(\sum_{i=1}^n  \calR^+(\hat v_i) - \sum_{i=1}^n \hat v_{i} \right)_{\pi^{-1}(j)}.
    \end{align*}
    As $\hat v_1,\dots,\hat v_n$ are uniformly random vectors from $\{-1,+1\}^d/\sqrt{d}$, we get that $\left(\sum_{i=1}^n \hat \calR^+(v_i) - \sum_{i=1}^n v_{i} \right)_j^2$ have the same distribution for all $j \in [d]$. The claim now follows since 
    \begin{align*}
    \Ex{\sum_{j=1}^d \left(\sum_{i=1}^n \hat \calR^+(v_i) - \sum_{i=1}^n v_{i} \right)_j^2}
    & = \Ex{\sum_{j=1}^d \left(\sum_{i=1}^n  \calR^+(\hat v_i) - \sum_{i=1}^n \hat v_{i} \right)_{\pi^{-1}(j)}^2} \\
    & = \Ex{\ltwo{\sum_{i=1}^n  \calR^+(\hat v_i) - \sum_{i=1}^n \hat v_{i}}^2 } \\
    & \le \Err(\calA^+,\calR).
    \end{align*}
\end{proof}

We are now ready to prove~\Cref{prop:reduc-small-to-large-dim}.
\begin{proof}(of~\Cref{prop:reduc-small-to-large-dim})
    $\calR'$ will work as follows for a $d'$-dimensional input $v'$: first, apply Kashin representation $U \in \reals^{2d' \times d'}$ to get $w = U v' \in \reals^{2d'}$ such that $\norm{w}_\infty \le 2/\sqrt{d'}$. Then, we convert $w$ into a binary vector $u$ by setting for all $i \in [2d']$
    \begin{equation*}
        u_i = 
        \begin{cases}
        2\mathsf{sign}(w_i)/\sqrt{d'} & \text{ with probability } \frac{w_i \sqrt{d'} + 2}{4} \\
        -2\mathsf{sign}(w_i)/\sqrt{d'} & \text{ with probability } \frac{-w_i \sqrt{d'} + 2}{4} 
        \end{cases}
    \end{equation*}
    Note that $E[u_i] = w_i $ and that $E[(w_i - u_i)^2] \le 4/d'$ since $|u_i| \le 2/\sqrt{d'}$.

    Now, let $\hat \calR$ be the randomizer guaranteed from~\Cref{lemma:sym-err-rand} for the randomizer $\calR$. Our local ranodmizer $\calR'$ will do the following: it constructs $v \in \reals^d$ by setting $v = (u,0,\dots,0)$ then applies $\hat \calR$ over $v$ to generate $k$ messages $m_1,\dots,m_k$. Finally, it truncates the messages to the first $2d'$ coordinates and applies the inverse Kashin transformation to the messages, that is, sends $U^T m_1[1:2d'],\dots, U^T m_k[1:2d']$
    
    Privacy of $\calR'$ follows immediately from privacy of $\calR$. It remains to prove an upper bound on the error for $\calR'$.

    Let $v'_1,\dots,v'_n \in \ball^{d'-1}$ and let $u_1,\dots,u_n$ be their corresponding binary vectors from the above procedure. Let $v_i = (u_i,0,\dots,0) \in \reals^d$. \Cref{lemma:sym-err-rand} guarantees that for all $j \in [d]$ we have 
    \begin{equation*}
          \Ex{\left\lvert\left(\sum_{i=1}^n \hat \calR^+(v_i) - v_{i} \right)_j\right\rvert^2} \le \frac{\Err(\calA,\calR)}{d}.
    \end{equation*}
    Thus, when truncating to the first $2d'$ coordinates of $\hat \calR$, we  have
    \begin{equation*}
          \Ex{\norm{\sum_{i=1}^n \hat \calR^+(v_i)[1:2d'] -  v_{i}[1:2d']}^2 } \le \frac{d'}{d} \Err(\calA,\calR) .
    \end{equation*}
    Now, let us analyze the error of $\calR'$. Note that 
    \begin{align*}
    \Ex{\norm{\sum_{i=1}^n \calR'^+(v'_i) -  v'_{i}}^2 }
        & = \Ex{\norm{\sum_{i=1}^n U^T \hat \calR^+(v_i)[1:2d']- v'_{i}}^2 } \\
        & = \Ex{\norm{\sum_{i=1}^n \hat \calR^+(v_i)[1:2d']- U v'_{i}}^2 } \\
        & \le \Ex{\norm{\sum_{i=1}^n \hat \calR^+(v_i)[1:2d'] - u_i + u_i -  w_i}^2 } \\
        & \le 2 \Ex{\norm{\sum_{i=1}^n   \hat \calR^+(v_i)[1:2d'] - u_i}^2}  + 2 \Ex{\norm{\sum_{i=1}^n u_i -  w_i}^2 } \\
        & = 2\Ex{\norm{\sum_{i=1}^n   \hat \calR^+(v_i)[1:2d'] - v_i[1:2d']}^2}  + 2 \Ex{\norm{\sum_{i=1}^n u_i -  w_i}^2} \\
        & \le 2 \frac{d'}{d} \Err(\calA,\calR)  + \frac{8n}{ d' } \\
        & \le \O{d'/\eps^2} , 
    \end{align*}
    where the last inequality follows since $\Err(\calA,\calR) \le \O{d/\eps^2}$ and $d' \ge n \eps^2/200$.

\end{proof}

\subsection{Missing Proofs from \texorpdfstring{\Cref{sec:agg-unbiased}}{Section 2.1.2}}

\subsubsection{Proof of \texorpdfstring{\Cref{prop:rec-attck-unb}}{Proposition 2.4}}
\label{apdx:proof-rec-attck-unb}
\proprecattackunb*
\begin{proof}
	The proof builds on the arguments of~\cite[Lemma 3.1]{AsiFeTa22} used in the local privacy model.  Let $P$ denote the uniform distribution over the sphere $\sphere^{d-1}$. First, note that as~\Cref{alg:atck-unb} iterates over all possible subsets of messages of size $k$, we have that $W_t = \calR(v_1)$ for some $t$, hence the set $S$ has the point
	\begin{equation*}
	u_t = \E_{\tilde v_2,\dots, \tilde v_{n} \sim P} \left[ \calA\left(\Pi(\calR(v_1), \calR(\tilde  v_2),\dots,\calR(\tilde  v_{n})) \right) \right] \in S
	\end{equation*}
	   We define $\hat \calR_i$ to be
    \begin{equation*}
        \hat \calR_i(v_i) = \E_{v_j \sim P, j \neq i} [\calA(\Pi(\calR(v_1),\dots,\calR(v_n)))] .
    \end{equation*}
    Note that $\hat \calR_1(v_1) \in S$ and that  $\E[\hat \calR_i(v)] = v$ for all $v \in \sphere^{d-1}$. We define 
    \begin{equation*}
        \hat \calR_{\le i}(v_1,\dots,v_i) = \E_{v_j \sim P, j > i} \left[\calA(\Pi(\calR(v_1),\dots,\calR(v_n))) - \sum_{j=1}^i v_j \mid v_{1:i}\right],
    \end{equation*}
    and $\hat \calR_0 = 0$.
    We now have
    \begin{align*}
     & \E_{v_1,\dots,v_n \sim P} \left[ \ltwo{\calA(\Pi(\calR(v_1),\dots,\calR(v_n))) -  \sum_{i=1}^n v_i}^2 \right]  \\
        & \quad = \E_{v_1,\dots,v_n \sim P} \left[ \ltwo{ \hat \calR_{\le n}(v_1,\dots,v_n)   }^2 \right] \\
        &  \quad = \E_{v_1,\dots,v_n \sim P} \left[ \ltwo{ \hat \calR_{\le n}(v_1,\dots,v_n) - \hat \calR_{\le n-1}(v_1,\dots,v_{n-1}) + \hat \calR_{\le n-1}(v_1,\dots,v_{n-1})    }^2 \right] \\
        & \quad  \stackrel{(i)}{=} \E_{v_1,\dots,v_n \sim P} \left[ \ltwo{ \hat \calR_{\le n}(v_1,\dots,v_n) - \hat \calR_{\le n-1}(v_1,\dots,v_{n-1}) }^2 \right] +  \E_{v_1,\dots,v_{n-1} \sim P} \left[ \ltwo{\hat \calR_{\le n-1}(v_1,\dots,v_{n-1})    }^2 \right] \\
        & \quad  \stackrel{(ii)}{=} \sum_{i=1}^n \E_{v_1,\dots,v_i \sim P} \left[ \ltwo{ \hat \calR_{\le i}(v_1,\dots,v_i) - \hat \calR_{\le i-1}(v_1,\dots,v_{i-1}) }^2 \right] \\
        & \quad  \stackrel{(iii)}{\ge} \sum_{i=1}^n \E_{v_i \sim P} \left[ \ltwo{ E_{v_1,\dots,v_{i-1} \sim P} \left[ \hat \calR_{\le i}(v_1,\dots,v_i) - \hat \calR_{\le i-1}(v_1,\dots,v_{i-1})\right] }^2 \right] \\
        & \quad  \stackrel{(iv)}{=} \sum_{i=1}^n \E_{v_i \sim P} \left[ \ltwo{\hat \calR_i(v_i) - v_i}^2 \right] \\
    \end{align*}
    where $(i)$ follows since $\E_{v_n \sim P} [\hat \calR_{\le n}(v_1,\dots,v_n)] = \hat \calR_{\le n-1}(v_1,\dots,v_{n-1})$, $(ii)$ follows by induction, $(iii)$ follows from Jensen's inequality, and $(iv)$ follows since $E_{v_1,\dots,v_{i-1} \sim P} [ \hat \calR_{\le i}(v_1,\dots,v_i)] = \hat \calR_i(v_i) - v_i$ and $E_{v_1,\dots,v_{i-1} \sim P} [ \hat \calR_{\le i-1}(v_1,\dots,v_{i-1})] = 0$.
    
   Now, as $\hat \calR_i$ has the same distribution for all $i$ because of the shuffling operator, we get that 
   \begin{align*}
    \E_{v_i \sim P} \left[ \ltwo{\hat \calR_1(v_1) - v_1}^2 \right] \le \Err(\calA,\calR)/n.
   \end{align*}
   Thus, as $\dist(v_1,S) \le \ltwo{\hat \calR_1(v_1) - v_1}^2$, the claim follows.
\end{proof}

\subsubsection{Proof of \texorpdfstring{\Cref{thm:main-lb-unb}}{Theorem 2.5}}
\label{apdx:main-lb-unb}
\thmmainlbunb*
\begin{proof}
    The proof will follow the proof of~\Cref{thm:main-lb} using the new reconstruction attack of~\Cref{alg:atck-unb}. Let $\Err(\calA,\calR) \le C \cdot d/\eps^2$ for some universal constant $1 \le C < \infty$.   

    First, we prove the lower bounds for $d \le n \eps^2/100C $. 
    Note that~\Cref{prop:rec-attck-unb} and Markov inequality imply that there is a set $A \subset \sphere^{d-1}$ such that $\PPPr{v \sim \mathsf{Unif}(\sphere^{d-1})}{A} \ge 1/2$ and for all $v \in A$ and $v_2,\dots,v_n \in \sphere^{d-1}$, letting $S_v$ be the output of~\Cref{alg:atck-unb} over the input $ \Pi(R(v), R(v_2),\dots,R(v_n))$, Markov inequality implies 
    \begin{equation*}
    	\PPr{\dist(v,S_v) \le 4d/n\eps^2} \ge 1/2.
    \end{equation*}
    As $\PPPr{v \sim \mathsf{Unif}(\sphere^{d-1})}{A} \ge 1/2$, this implies that there is a $\rho$-packing of the unit ball $P = \{v_1,v_2,\dots,v_M \} \subset A$  such that $M = 2^{d \log(1/\rho) - 1}$ and $\Pr(\dist(v_i,S_{v_i}) \le 4 C d /n \eps^2 ) \ge 1/2$.

    We will prove the lower bounds by analyzing the algorithm over the following $M$ datasets:
    \begin{equation*}
        X_i = ( v_i,v_1,\dots,v_1),
    \end{equation*}
    for $i \in [M]$.
    
Let $S_i$ be the output of the reconstruction attack (\Cref{alg:atck-unb}) over the shuffled messages $\Pi(\calR(v_i),\calR(v_1),\dots,\calR(v_1))$. We define the projection set of $S_i$ to the packing $P$ to be $O_i = \{ \mathsf{Proj}_P(v) : v \in S_i\}$.
    \Cref{prop:rec-attck-unb} now implies that for all $i \in [M]$, $\dist(v_i,S_i) \le  Cd /n \eps^2 \le \rho$ with probability $1/2$, hence as $P$ is $\rho$-packing we have that 
    \begin{equation*}
        \PPr{v_i \in O_i} \ge \PPr{\dist(v_i,S_i} \le \rho ) \ge  9/10.
    \end{equation*}

    On the other hand, note that for $O_1$
    \begin{align*}
    \sum_{i =1}^M  \PPr{v_i \in O_1}
        & = \sum_{i =1}^M\Ex{1\{v_i \in O_1\}}  \\
        & = \Ex{\sum_{i =1}^M 1\{v_i \in O_1\}}  \\
        & \le \Ex{|O_1|} \le \binom{nk}{k}.
    \end{align*}
    Hence there exists an $1 \le i \le M$ such that 
    \begin{equation*}
         \PPr{v_i \in O_1} \le \frac{\binom{nk}{k}}{M}.
    \end{equation*}
    As the protocol is $(\eps,\delta)$-DP, we also have 
    \begin{align*}
     \PPr{v_i \in O_1} 
        & \ge\PPr{v_i \in O_i} e^{-\eps} - \delta \\
        & \ge \frac{9}{10 e} - \delta \ge 1/6
    \end{align*}
    Combining these together, and given that $M \ge 2^d/2$ for $\rho = 1/10$, we have that 
    \begin{equation*}
        2^d \le 12 \binom{nk}{k} \le 6 (en)^k.
    \end{equation*}
    This implies that $k \ge \Omega(d/\log(n))$ whenever $d \le n\eps^2/100C$.
    
    Now we prove the lower bound for $d \ge n \eps^2/100C$. 
    The proof builds on the following proposition which states that we can convert an optimal protocol for $d$-dimensional inputs into an optimal protocol for $d'$-dimensional inputs where $ d' = n \eps^2/200C$ with the same number of messages. 
    We defer the proof to~\Cref{sec:proof-prop-red-small-to-large-dim-unbiased}.
    \begin{proposition}
    \label{prop:reduc-small-to-large-dim-unbiased}
    Let $ d' = n \eps^2/200C \ge 1 $ and $d \ge 2d'$.
    Let $\calR : \sphere^{d-1} \to \calZ^k$ be an $(\eps,\delta)$-Shuffle DP randomizer with aggregation $\calA$ that is unbiased such that $\Err(\calA,\calR) \le \O{d/\eps^2}$. There exists $\calR' : \sphere^{d'-1} \to \calZ^k$ and aggregation $\calA'$ that is unbiased and $(\eps,\delta)$-Shuffle DP such that $\Err(\calA',\calR') \le  \O{d'/\eps^2} $.
    \end{proposition}
    Let $\calA$ and $\calR : \sphere^{d-1} \to \calZ^k$ be an unbiased $(\eps,\delta)$-Shuffle DP protocol that obtains error $\Err(\calA, \calR) \le\O{d/\eps^2}$ using $k$ messages. \Cref{prop:reduc-small-to-large-dim-unbiased} implies that there is a randomizer $\calR': \sphere^{d'-1} \to \calZ^k$ and aggregation $\calA'$ that is $(\eps,\delta)$-Shuffle DP and unbiased such that $\Err(\calA',\calR') \le \O{d'/\eps^2}$ for $d' = n \eps^2/200C$. As $d' \le n\eps^2/100C$, the lower bound we proved above shows that $k \ge \Omega(d'/\log(n)) =  \Omega(n \eps^2/\log(n))$.
\end{proof}

\subsubsection{Proof of \texorpdfstring{\Cref{prop:reduc-small-to-large-dim-unbiased}}{Proposition A.3}}
\label{sec:proof-prop-red-small-to-large-dim-unbiased}
The proof will follow the proof of~\Cref{sec:proof-prop-red-small-to-large-dim} with general aggregation $\calA$. To this end, in the next lemma we show that we can convert any unbiased protocol into another unbiased one where the error is split evenly across coordinates.
\begin{lemma}
\label{lemma:sym-err-rand-unbiased}
    If $\calR: \sphere^{d-1} \to \calZ^k$ is $(\eps,\delta)$-shuffle DP randomizer and $\calA$ is unbiased, then there exists $ \calR' : \left\{\frac{-1}{\sqrt{d}},\frac{1}{\sqrt{d}} \right\}^d \to \calZ^k $ and $\calA'$ that is $(\eps,\delta)$-shuffle DP and unbiased such that for $j \in [d]$ and $v_1,\dots,v_n \in \left\{ \frac{-1}{\sqrt{d}}, \frac{+1}{\sqrt{d}}\right\}^d$, 
    \begin{equation*}
          \E\left[ \left|\left(\calA'(\Pi(\calR'(v_1),\dots,\calR'(v_n))) - \sum_{i=1}^n v_i \right)_j \right|^2 \right] \le \frac{\Err(\calA,\calR)}{d}.
    \end{equation*}
\end{lemma}

\begin{proof}
    $\calR'$ will use shared public randomness to shuffle the coordinates of each vector and flip the signs of each coordinate. This will ensure that all coordinates will have the same marginal distribution for their error. 

    More precisely, let $\pi: [d] \to [d] $ be a random permutation of the coordinates picked uniformly at random, and let $s_1,\dots,s_d \sim \mathsf{Ber}(1/2)$. Our new randomizer $\calR'$ over input $v$ has
    \begin{equation*}
        \calR'(v) =\calR(s \cdot v(\pi)),
    \end{equation*}
    where $(s\cdot v(\pi))_j = s_j v_{\pi(j)}$ is element-wise product.

    Moreover, we define $\calA'$ given $kn$ messages $m_i \in \calZ$ 
    \begin{equation*}
    \calA'(m_1,\dots,m_{kn}) = s \cdot \calA(m_1,\dots,m_{kn})(\pi^{-1}).
    \end{equation*}
    First, note that the privacy of $\calR'$ follows immediately from the privacy of $\calR$. Moreover, $\calA'$ is unbiased as $\calA$ is unbiased: 
    \begin{align*}
    \Ex{\calA'(\Pi(\calR'(v_1),\dots,\calR'(v_n))) }
        & = s \cdot \Ex{\calA(\Pi(\calR(s \cdot v_1(\pi)),\dots,\calR(s \cdot v_n(\pi))))(\pi^{-1}) } \\
        & = s \cdot \sum_{i=1}^n s \cdot v_i (\pi) (\pi^{-1}) \\
        & = \sum_{i=1}^n  v_i ,
    \end{align*}
    where the last equality follows since $s \cdot s = 1^d$ and $v_i(\pi)(\pi^{-1}) = v_i$.

    Now it remains to prove the claim about the error of $\calR'$ and $\calA'$. Letting $\hat v_i = s \cdot v_i(\pi)$, note that $v = s \cdot \hat v(\pi^{-1})$, thus we get
    \begin{align}
    \left| \left(\calA'(\Pi(\calR'(v_1),\dots,\calR'(v_n))) - \sum_{i=1}^n v_{i}\right)_j \right|^2
        & =\left| \left( s \cdot  \calA(\Pi(\calR(s \cdot v_1(\pi)),\dots,\calR(s \cdot v_n(\pi))))(\pi^{-1}) - \sum_{i=1}^n v_{i} \right)_j \right|^2 \nonumber \\
        & = \left| \left( s \cdot \calA(\Pi(\calR(\hat v_1),\dots,\calR(\hat v_n)))(\pi^{-1}) - \sum_{i=1}^n s \cdot \hat v_i(\pi^{-1}) \right)_j \right|^2 \nonumber \\
        & = \left|  \left( \calA(\Pi(\calR(\hat v_1),\dots,\calR(\hat v_n))) - \sum_{i=1}^n  \hat v_{i} \right)_{\pi^{-1}(j)}\right|^2 \label{eq:same-dist-unbiased} 
    \end{align}
    Summing over all coordinates, 
    \begin{align*}
    \E\left[\sum_{j=1}^d \left(\calA'(\Pi(\calR'(v_1),\dots,\calR'(v_n))) - \sum_{i=1}^n v_{i} \right)_j^2 \right]
    & = \E\left[\sum_{j=1}^d \left(\sum_{i=1}^n  \calA(\Pi(\calR(\hat v_1),\dots,\calR(\hat v_n))) - \sum_{i=1}^n \hat v_{i} \right)_{\pi^{-1}(j)}^2 \right] \\
    & = \E\left[\ltwo{\calA(\Pi(\calR(\hat v_1),\dots,\calR(\hat v_n))) - \sum_{i=1}^n \hat v_{i}}^2 \right]  \\
    & \le \Err(\calA,\calR).
    \end{align*}
    Finally, the claim now follows since for all $j \in [d]$, $| (\calA'(\Pi(\calR'(v_1),\dots,\calR'(v_n))) - \sum_{i=1}^n v_{i})_j |^2$ have the same distribution and hence the same expectation: indeed, let 
    \begin{equation*}
        A = | ( \calA(\Pi(\calR(\hat v_1),\dots,\calR(\hat v_n))) - \sum_{i=1}^n  \hat v_{i} )|^2 \quad \text{  and } t = \pi^{-1}(j).
    \end{equation*}
    Equation~\eqref{eq:same-dist-unbiased} shows that $| (\calA'(\Pi(\calR'(v_1),\dots,\calR'(v_n))) - \sum_{i=1}^n v_{i})_j |^2 = A_t$. Now note that $\hat v_1,\dots,\hat v_n$ are uniformly random vectors from $\{-1,+1\}^d/\sqrt{d}$ and $\pi^{-1}(j)$ is random coordinate from $[d]$, hence the distribution of $A_t$ is the same for all $j$.
\end{proof}

We are now ready to prove~\Cref{prop:reduc-small-to-large-dim-unbiased} which will follow the proof of~\Cref{prop:reduc-small-to-large-dim}.
\begin{proof}(of~\Cref{prop:reduc-small-to-large-dim-unbiased}) 
    We will construct $\calR'$ and $\calA'$ as follows:
    for a $d'$-dimensional input $v' \in \sphere^{d'-1}$, $\calR'$ will first apply the Kashin representation $U \in \reals^{2d' \times d'}$ to get $w = U v' \in \reals^{2d'}$ such that $\norm{w}_\infty \le 2/\sqrt{d'}$. Then, it converts $w$ into a binary vector $u$ by setting for all $i \in [2d']$
    \begin{equation*}
        u_i = 
        \begin{cases}
        2\mathsf{sign}(w_i)/\sqrt{d'} & \text{ with probability } \frac{w_i \sqrt{d'} + 2}{4} \\
        -2\mathsf{sign}(w_i)/\sqrt{d'} & \text{ with probability } \frac{-w_i \sqrt{d'} + 2}{4} 
        \end{cases}
    \end{equation*}
    Note that $E[u_i] = w_i $ and that $E[(w_i - u_i)^2] \le 4/d'$ since $|u_i| \le 2/\sqrt{d'}$.

    Now, let $\hat \calR$ and $\hat \calA$ be the randomizer and aggregation guaranteed from~\Cref{lemma:sym-err-rand} for the randomizer $\calR$ and aggregation $\calA$. Our $\calR'$ will  construct $v \in \reals^d$ by setting $v = (u,0,\dots,0)$ then  
    \begin{equation*}
       \calR'(v') = \hat \calR(v).
    \end{equation*}
    Moreover, we define $\calA': \calZ^{nk} \to \reals^{d'}$ to be 
    \begin{equation*}
     \calA'(m_1,\dots,m_{nk}) =  U^T \hat \calA(m_1,\dots,m_{nk})[1:2d'] 
    \end{equation*}

    We nee to argue that $\calR'$ is $(\eps,\delta)$-Shuffle DP, that $\calA'$ is unbiased, and to prove the claim about utility.
    
    Privacy of $\calR'$ follows immediately from privacy of $\hat \calR$. As for unbiasedness,  let $v'_1,\dots,v'_n \in \sphere^{d'-1}$ and let $u_1,\dots,u_n$ and $w_1,\dots,w_n$ be their corresponding vectors from the above procedure. Let $v_i = (u_i,0,\dots,0) \in \reals^d$. Note that 
    \begin{align*}
    \Ex{\calA'(\Pi(\calR'( v'_1),\dots,\calR'( v'_n)))}
        & = \Ex{U^T \hat \calA(\Pi(\hat \calR( v_1),\dots,\hat \calR( v_n)))[1:2d']} \\ 
        & = U^T \sum_{i=1}^n \Ex{ v_i[1:2d']} \\ 
        & = U^T \sum_{i=1}^n \Ex{ u_i} \\ 
        & = U^T \sum_{i=1}^n  w_i \\ 
        & = U^T \sum_{i=1}^n U v'_i \\ 
        & =  \sum_{i=1}^n  v'_i.
    \end{align*}
    It remains to prove an upper bound on the error of $\calR'$ and $\calA'$. First, note that \Cref{lemma:sym-err-rand} guarantees that for all $j \in [d]$ we have 
    \begin{equation*}
          E[|\left(\hat \calA(\Pi(\hat \calR( v_1),\dots,\hat \calR( v_n))) - \sum_{i=1}^n v_{i} \right)_j|^2 ] \le \frac{\Err(\calA,\calR)}{d}.
    \end{equation*}
    Thus, when truncating to the first $2d'$ coordinates of $\hat \calR$, we  have
    \begin{equation*}
          E[\norm{\hat \calA(\Pi(\hat \calR( v_1),\dots,\hat \calR( v_n)))[1:2d'] -  \sum_{i=1}^n  v_{i}[1:2d']}^2 ] \le \frac{2d'}{d} \Err(\calA,\calR) .
    \end{equation*}
    Now, let us analyze the error of $\calR'$. Note that 
    \begin{align*}
    &\Ex{\norm{\calA'(\Pi(\calR'( v'_1),\dots,\calR'( v'_n))) - \sum_{i=1}^n v'_{i}}^2} \\
        & \quad = \Ex{\norm{U^T \hat \calA(\Pi(\hat \calR( v_1),\dots,\hat \calR( v_n)))[1:2d'] - \sum_{i=1}^n v'_{i}}^2} \\
        & \quad  = \Ex{\norm{\hat \calA(\Pi(\hat \calR( v_1),\dots,\hat \calR( v_n)))[1:2d'] - \sum_{i=1}^n U v'_{i}}^2} \\
        & \quad  \le \Ex{\norm{\hat \calA(\Pi(\hat \calR( v_1),\dots,\hat \calR( v_n)))[1:2d'] - \sum_{i=1}^n u_i - u_i +  w_i}^2} \\
        & \quad  \le 2 \Ex{\norm{\hat \calA(\Pi(\hat \calR( v_1),\dots,\hat \calR( v_n)))[1:2d'] - \sum_{i=1}^n u_i}^2}  + 2 \Ex{\norm{\sum_{i=1}^n u_i -  w_i}^2} \\
        & \quad  = 2 \Ex{\norm{\hat \calA(\Pi(\hat \calR( v_1),\dots,\hat \calR( v_n)))[1:2d'] - \sum_{i=1}^n v_i[1:2d']}^2}  + 2 \Ex{\norm{\sum_{i=1}^n u_i -  w_i}^2} \\
        & \quad  \le 4 \frac{d'}{d} \Err(\calA,\calR)  + \frac{8n}{ d' } \\
        & \quad  \le \O{d'/\eps^2} , 
    \end{align*}
    where the last inequality follows since $\Err(\calA,\calR) \le \O{d/\eps^2}$ and $d' \ge n \eps^2/200$.

\end{proof}

\subsection{Missing Proofs for \texorpdfstring{\Cref{sec:opt:mult}}{Section 2.2}}
\label{sec:proofs-ub-multi}

Our $d$-dimensional algorithm builds on the $1$-dimensional algorithm by \cite{GhaziKMPS21}. We let $\Rg^{(\eps,\delta)}$ denote the local randomizer with parameters $(\eps,\delta)$ and $\calA^+$ is their aggregation (which is summation over messages). Their protocol has the following guarantees for $1$-dimensional summation.

\begin{lemma}
\label{lemma:ghazi}
\cite{GhaziKMPS21}
There is a local randomizer $\Rg^{(\eps,\delta)}: [0,1] \to \reals^\star$ that is $(\eps,\delta)$-Shuffle DP such that each user sends $1 + \widetilde{\mathcal{O}}_{\eps}\left(\frac{\log(1/\delta)}{\sqrt{n}}\right)$ in expectation and has error
\begin{equation*}
    \Err(\Rg^{(\eps,\delta)},\calA^+)  \le \O{1/\eps^2}.
\end{equation*}
\end{lemma}

\noindent

We also used advanced composition in our privacy proof.
\begin{lemma}[Advanced composition~\cite{DworkR14}] 
\label{lemma:advanced-comp}
    If $A_1,\dots,A_k$ are randomized algorithms that each is $(\eps,\delta)$-DP, then their composition $(A_1(D),\dots,A_k(D))$ is $(\sqrt{2k \log(1/\delta')} \eps + k \eps (e^\eps - 1),\delta' + k \delta)$-DP where $D$ is the input dataset.
\end{lemma}

Now we present the guarantees of our protocol. 
\thmmultmsgerrbound*
\begin{proof}
    First, note that the guarantees of Kashin representation imply that each $|u^{(i)}_j| \le 1$, hence we can use $\Rg$.
    The claim from privacy follows from the fact the $\Rg$ is $(\eps_0,\delta_0)$-Shuffle DP and advanced composition of $d$ such mechanisms. 
    
    Now we analyze the error of the protocol. We have
    \begin{align*}
    \Ex{\ltwo{\hat v - \sum_{i=1}^n v^{(i)}}^2} 
    & = \Ex{ \ltwo{\frac{\Ck}{\sqrt{d}}  U^T_{\mathsf{K}} \hat u - \frac{\Ck}{\sqrt{d}}  U^T_{\mathsf{K}} \sum_{i=1}^n u^{(i)}}^2} \\
    & = \frac{\Ck^2}{d} \Ex{\ltwo{ \hat u -  \sum_{i=1}^n u^{(i)}}^2} \\
    & = \frac{\Ck^2}{d} \Ex{ \ltwo{\sum_{m \in M} m  - u^{(i)}}^2} \\
    & \le  \frac{\Ck^2}{d} \frac{d}{\eps_0^2} \\ 
    & \le \O{\frac{d \log(1/\delta)}{\eps^2}},
    \end{align*}
    where the last inequality follows from the guarantees of the $\Rg$ protocol which has error $1/\eps_0^2$ in each coordinate. The claim follows.

\end{proof}

\section{Missing Proofs from \texorpdfstring{\Cref{sec:single-msg}}{Section 3}}
\applab{app:single:msg}
To prove the lower bound, we first note that it suffices to assume that the local randomizer has bounded outputs and that the analyzer simply adds up all of the messages sent by the users, as shown by the next lemma.
\begin{restatable}{lemma}{lemlbreduction}
\lemlab{lem:lb:reduction}
Let $\calP=(\calR,\calA)$ be an $n$-party protocol for vector aggregation in the single-message shuffle model. 
Let $V$ be a random variable on $\left[-\frac{1}{\sqrt{d}},\frac{1}{\sqrt{d}}\right]^d$ and suppose that users sample their inputs from the distribution $V^n$. 
Then there exists a protocol $\calP'=(\calR',\calA')$ with user outputs $u_1,\ldots,u_n\in\mathbb{R}^d$ such that:
\begin{enumerate}
\item 
$\calA'(u_1,\ldots,u_n)=\sum_{i=1}^n u_i$ and $\calR'$ maps to $\left[-\frac{1}{\sqrt{d}},\frac{1}{\sqrt{d}}\right]^d$. 
\item 
$\MSE(\calP',V)\le\MSE(\calP,V)$
\item 
If $\calS\circ\calR^n$ is $(\eps,\delta)$-DP, then $\calS\circ(\calR')^n$ is $(\eps,\delta)$-DP.
\end{enumerate}
\end{restatable}
\begin{proof}
The proof is similar to Lemma 4.1 in~\cite{BalleBGN19}, generalizing from scalars to vectors. 
Let $\calR'=f\circ\calR$ be the post-processing local randomizer that uses the posterior mean estimator $f(u)=\Ex{V\,\mid\,\calV=u}$ is the minimum MSE estimator. 
Then $\calR'$ maps to $\left[-\frac{1}{\sqrt{d}},\frac{1}{\sqrt{d}}\right]^d$ as claimed.

Observe that for any estimator $h$ of $Z:=V_1+\ldots+V_n$ given the input $U=\{u_1,\ldots,u_n\}$, we have
\begin{align*}
\MSE(h,U)&=\Ex{(h(u)-Z)^2\,\mid\,U}\\
&=\Ex{Z^2\,\mid\,U}-2h(u)\cdot\Ex{Z\,\mid\,U}+(h(u))^2.
\end{align*}
This quantity is minimized over the choice of $h$ at $h(u)=\Ex{Z\,\mid\,U}$. 

Finally, since $f$ is a post-processing local randomizer, then $\calS\circ(\calR')^n$ is $(\eps,\delta)$-DP by the post-processing property of DP. 
\end{proof}

\begin{restatable}{lemma}{lemmsescales}
\lemlab{lem:mse:scales}
Let $\calP=(\calR,\calA)$ be an $n$-party protocol for vector aggregation in the single-message shuffle model such that $\calR:\left[-\frac{1}{\sqrt{d}},\frac{1}{\sqrt{d}}\right]^d\to\left[-\frac{1}{\sqrt{d}},\frac{1}{\sqrt{d}}\right]^d$ and $\calA$ is vector summation. 
Suppose $V^n$ are $n$ copies of a random variable $V$. 
Then
\[\MSE(\calP,V^n)\ge n\Ex{\|\calR(V)-V\|_2^2}.\]
\end{restatable}
\begin{proof}
The proof generalizes Lemma 4.2 in~\cite{BalleBGN19} from scalar inputs to vector inputs. 
Note that we can decompose the mean-squared error as follows. 
\begin{align*}
\MSE&(\calP,V^n)=\Ex{\|\sum_{i=1}\calR(V_i)-V_i\|_2^2}\\
&=\sum_i\Ex{\|\calR(V_i)-V_i\|_2^2}+\sum_{i\neq j}\Ex{\langle\calR(V_i)-V_i,\calR(V_j)-V_j\rangle}\|\\
&=\sum_i\Ex{\|\calR(V_i)-V_i\|_2^2}+\sum_{i\neq j}\langle\Ex{\calR(V_i)-V_i},\Ex{\calR(V_i)-V_i}\rangle\\
&\ge n\Ex{\|\calR(V)-V\|_2^2}.
\end{align*}
\end{proof}

Consider the partition $P$ of the hypercube $[0,1]^d$ into $r^d$ disjoint hypercubes with side length $\frac{1}{r}$. 
Let $I=\left\{\frac{m}{r}-\frac{1}{2r}\,\mid\,m\in[r]\right\}$ and $J=I^d$. 
For each $a\in J$, we use $J(a)$ to denote the hypercube of $P$ that contains $J$. 
For any $b\in J$, we use the notation $p_{a,b}$ to denote the probability that the randomizer maps $a$ to $I(b)$. 

\begin{restatable}{lemma}{lematob}
\lemlab{lem:atob}
Let $r\ge 32$. 
For any $b\in J$, we have
\[\frac{1}{r^d}\sum_{a\in J\setminus b}\left(\min\left(\|a-b\|_2-\frac{\sqrt{d}}{2r},0\right)\right)^2\ge\frac{d}{2048}.\]
\end{restatable}
\begin{proof}
Let $B$ be a hypercube with length $\frac{1}{8}$ centered at $b$. 
Note that we have $\PPr{a\in J\setminus B}\ge\frac{1}{2}$. 
For $a\in J\setminus B$, we have $\|a-b\|_2\ge\frac{\sqrt{d}}{16}$. 
Then for $r\ge 64$, we have $\left(\|a-b\|_2-\frac{1}{2r}\right)^2\ge\frac{d}{32^2}$. 
Hence we have 
\[\frac{1}{r^d}\sum_{a\in J\setminus b}\left(\|a-b\|_2-\frac{1}{2r}\right)^2\ge\frac{1}{2}\cdot\frac{d}{32^2}=\frac{d}{2048}.\]
\end{proof}

\begin{lemma}
\lemlab{lem:err:inblock}
The mean-squared error of the randomizer $\calR$ on the random variable $V$ is at least:
\[\Ex{\|\calR(V)-V\|_2^2}\ge\sum_{b\in J}\min\left(\frac{d(1-p_{b,b})}{4r^{2+d}},\min_{a\in J}p_{a,b}\cdot\frac{d}{2048}\right).\]
\end{lemma}
\begin{proof}
For the cases where the randomizer maps $V$ to a value outside of its hypercube, we have:
\begin{align*}
\Ex{\|\calR(V)-V\|_2^2}&=\sum_{b\in J}\Ex{\|\calR(b)-b\|_2^2}\cdot\PPr{V=b}\\
&=\frac{1}{r^d}\sum_{b\in J}\Ex{\|\calR(b)-b\|_2^2}\\
&\ge\frac{1}{r^d}\sum_{b\in J}(1-p_{b,b})\cdot\frac{d}{4r^2}\\
&=\sum_{b\in J}\frac{d(1-p_{b,b})}{4r^{2+d}}.
\end{align*}
We also have
\begin{align*}
\mathbb{E}&\left[\|\calR(V)-V\|_2^2\right]=\frac{1}{r^d}\sum_{b\in J}\Ex{\|\calR(b)-b\|_2^2}\\
&\ge\frac{1}{r^d}\sum_{b\in J}\sum_{a\in J\setminus b}p_{a,b}\left(\min\left(\|a-b\|_2-\frac{\sqrt{d}}{2r},0\right)\right)^2\\
&\ge\frac{1}{r^d}\sum_{b\in J}\min_{a\in J}p_{a,b}\sum_{a\in J\setminus b}\left(\min\left(\|a-b\|_2-\frac{\sqrt{d}}{2r},0\right)\right)^2\\
&\ge\sum_{b\in J}\min_{a\in J}p_{a,b}\cdot\frac{d}{2048},
\end{align*}
where the last inequality is from \lemref{lem:atob}. Hence, we have
\[\Ex{\|\calR(V)-V\|_2^2}\ge\sum_{b\in J}\min\left(\frac{d(1-p_{b,b})}{4r^{2+d}},\min_{a\in J}p_{a,b}\cdot\frac{d}{2048}\right).\]
\end{proof}

\begin{lemma}[Lemma 4.5 in \cite{BalleBGN19}]
\lemlab{lem:prob:cases}
Let $\calR:[0,1]^d\to[0,1]^d$ be a local randomizer such that the shuffled protocol $\calM=\calS\circ\calR^n$ is $(\eps,\delta)$-DP with $\delta<\frac{1}{2}$. 
Then for any $a,b\in J$ with $a\neq b$, we have either $p_{b,b,}<1-\frac{e^{-\eps}}{2}$ or $p_{a,b}\ge\frac{1}{n}\cdot\left(\frac{1}{2}-\delta\right)$. 
\end{lemma}

We are now ready to prove the lower bound.
\begin{proof}
By \lemref{lem:mse:scales}, we have $\MSE(\calP,V^n)\ge n\Ex{\|\calR(V)-V\|_2^2}$. 
By \lemref{lem:err:inblock}, we have $\Ex{\|\calR(V)-V\|_2^2}\ge\sum_{b\in J}\min\left(\frac{d(1-p_{b,b})}{4r^{2+d}},\min_{a\in J}p_{a,b}\cdot\frac{d}{2048}\right)$. 
Therefore by \lemref{lem:prob:cases}, 
\begin{align*}
\MSE(\calP,V)&\ge n\sum_{b\in J}\min\left(\frac{d(1-p_{b,b})}{4r^{2+d}},\min_{a\in J}p_{a,b}\cdot\frac{d}{2048}\right)\\
&\ge n\sum_{b\in J}\min\left(\frac{de^{-\eps}}{4r^{2+d}},\frac{1}{n}\cdot\left(\frac{1}{2}-\delta\right)\cdot\frac{d}{2048}\right)\\
&\ge nr^d\min\left(\frac{de^{-\eps}}{4r^{2+d}},\frac{1}{n}\cdot\left(\frac{1}{2}-\delta\right)\cdot\frac{d}{2048}\right).
\end{align*}
The quantity is maximized for $r=\O{n^{1/(d+2)}}$ with value $\Omega\left(dn^{d/(d+2)}\right)$. 
\end{proof}

\section{Missing Proofs from \texorpdfstring{\Cref{sec:rob}}{Section 4}}
\thmattackoneshuffler*
\begin{proof}
Suppose $\calA$ is a protocol in which 1) each user receives an input $v$ and outputs $d$ messages from a randomizer $\calR(v)$, 2) after shuffling, the protocol collects the messages and outputs the sum of the messages. 
We consider casework on the distribution of the output of the randomizer $\calR$. 

Firstly, suppose that for an input vector $v$, $\PPr{\max_{m\in\calR(v)}\|m\|_2\ge\frac{1}{\sqrt{d}\alpha}}\ge\frac{1}{dn^2}$, for some parameter $\alpha>1$ to be fixed.  
Note that a single malicious user can then run the randomizer $\calR$ on inputs $v^{(1)}$ and $v^{(2)}$ a total of $\O{dn^2}$ times and with probability $0.99$, find a message $m$ such that $\|m\|_2\ge\frac{1}{\sqrt{d}\alpha}$. 
Note that the malicious user can send the message $m$ a total of $d$ times, which contributes $L_2$ norm $\frac{\sqrt{d}}{\alpha}$. 
Since each malicious user previously had a unit vector, then the mean squared error induced by each malicious user is at least $\frac{d}{\alpha^2}$. 
Therefore, $k$ malicious users can induce mean squared error $\frac{kd}{\alpha^2}$. 

Secondly, suppose that for an input vector $v$, we have $\sup\left\langle\frac{m}{\|m\|_2},v\right\rangle>\frac{100\sqrt{\log nd}}{\sqrt{d}}$. 
We claim this would violate privacy. 
Note that for a random vector $u$, we have by the rotational invariance of Gaussians,
\[\PPr{\left\langle\frac{m}{\|m\|_2},u\right\rangle>\frac{100\sqrt{\log nd}}{\sqrt{d}}}<\frac{1}{10n^2d^2}.\]
With probability at least $\frac{1}{10nd}$, none of the $nd$ messages has correlation at least $\frac{100\sqrt{\log nd}}{d}$ with $u$. 
Thus we would be able to distinguish between the cases where the inputs are the neighboring datasets $(v,v,\ldots,v)$ and $(u,v,\ldots,v)$, which contradicts $(\eps,\delta)$-differential privacy for $\eps=\O{1}$ and $\delta<\frac{1}{nd}$, 

It remains to consider the case where $\max_{m\in\calR(v)\cup\calR(u)}\|m\|_2<\frac{1}{\sqrt{d}\alpha}$ and $\sup\left\langle\frac{m}{\|m\|_2},v\right\rangle\le\frac{100\sqrt{\log nd}}{\sqrt{d}}$. 
Note that in this case, we have
\begin{align*}
\left\langle\sum_{i\in[d]} m_i,v\right\rangle&=\sum_{i\in[d]}\|m_i\|_2\cdot\left\langle\frac{m_i}{\|m_i\|_2},v\right\rangle\\
&\le\sup_{i\in[d]}\|m_i\|_2\cdot d\cdot\sup_{i\in[d]}\left\langle\frac{m_i}{\|m_i\|_2},v\right\rangle\\
&\le\frac{100\sqrt{\log nd}}{\sqrt{d}}\cdot d\cdot\frac{1}{\sqrt{d}\alpha}=\frac{100\sqrt{\log nd}}{\alpha}.
\end{align*}
Thus for $\alpha>200\sqrt{\log nd}$ and an elementary vector $v$, we have that
\[\left\langle\sum_{i\in[d]} m_i,v\right\rangle\le\frac{1}{2},\]
and thus the mean squared error for the input $(v,v,\ldots,v)$ would be at least $\frac{n}{2}$. 

Hence for $n>kd$, the mean squared error induced by $k$ malicious users is at least $\Omega\left(\frac{kd}{\log^2(nd)}\right)$. 
\end{proof}

\thmbaduserserr*
\begin{proof}
For $i\in[s]$, let $d_i$ be the number of coordinates for which the $i$-th shuffler is responsible. 
Then we have $d_1+\ldots+d_s$. 
By \thmref{thm:attack:one:shuffler}, there exists a set of messages for which $k$ malicious users can induce mean squared error $\Omega\left(\frac{kd_i^2}{\log^2(nd)}\right)$ through sum of the messages in the $i$-th shuffler. 
%
%
Now, we have that the mean squared error is $\sum_{j\in[n]}\|x_j\|_2^2$
$C\left(\sum_{i\in[s]}\frac{kd_i^2}{\log^2(nd)}\right)$, which is minimized at $\Omega\left(\frac{kd^2}{s\log^2(nd)}\right)$ for $d_1=\ldots=d_s=\frac{d}{s}$ by a standard power mean inequality. 
\end{proof}

\end{document}